%% file: main.tex
\documentclass[11pt]{article}
\usepackage{times}
\usepackage{amsmath,amssymb,amsthm,amsfonts,latexsym,bbm,xspace,graphicx,float,mathtools,dsfont}
\usepackage[backref,colorlinks,citecolor=blue,bookmarks=true]{hyperref}
\usepackage[dvips, paper=letterpaper, top=1in, bottom=.75in, left=1in, right=1in, nohead, includefoot, footskip=.25in]{geometry}
\usepackage{color}
\usepackage{comment}
\usepackage{algorithm}
\usepackage{algorithmic}

\usepackage{tweaklist}
\renewcommand{\enumhook}{\setlength{\topsep}{2pt}%
  \setlength{\itemsep}{-2pt}}
\renewcommand{\itemhook}{\setlength{\topsep}{2pt}%
  \setlength{\itemsep}{-2pt}}

\usepackage{fullpage}
\usepackage{appendix}
\usepackage{cleveref}

\usepackage[nomessages]{fp}

\newtheorem{theorem}{Theorem}

\newtheorem{corollary}[theorem]{Corollary}
\newtheorem{lemma}[theorem]{Lemma}

\newtheorem{claim}[theorem]{Claim}
\newtheorem{proposition}[theorem]{Proposition}

\newtheorem{definition}[theorem]{Definition}

\newtheorem{observation}[theorem]{Observation}

\newcommand{\rev}{\textsc{Rev}}
\newcommand{\bvcg}{\textsc{BVCG}}
\newcommand{\copies}{$\textsc{OPT}^{\textsc{Copies}}$}
\newcommand{\vcg}{\textsc{VCG}}
\newcommand{\srev}{\textsc{SRev}}
\newcommand{\brev}{\textsc{BRev}}

\newenvironment{prevproof}[2]{\noindent {\em {Proof of {#1}~\ref{#2}:}}}{$\Box$\vskip \belowdisplayskip}

\newcommand{\junk}[1]{}

\junk{

}

\makeatother

\newcommand{\notshow}[1]{{}}

\DeclareMathOperator{\Prob}{Pr}
\DeclareMathOperator{\E}{E}

\def \E  {{\mathbb{E}}}
\def \Var{{\text{Var}}}
\def \polytope{{P(\mathcal{F},D)}}
\def \L{{\mathcal{L}}}

\DeclareMathOperator{\argmax}{argmax}

\def\tp{{\tilde{\varphi}}}
\def\I{{\mathds{1}}}
\definecolor{MyGray}{rgb}{0.8,0.8,0.8}

\begin{document}

\title{A Duality-Based Unified Approach to Bayesian Mechanism Design}

\author{
Yang Cai\thanks{Supported by NSERC Discovery RGPIN-2015-06127. Work done in part while the author was a Research Fellow at the Simons Institute for the Theory of Computing.} \\
McGill University, Canada\\
cai@cs.mcgill.ca \and Nikhil R. Devanur\thanks{Work done in part while the author was visiting the Simons Institute for the Theory of Computing.} \\Microsoft Resarch, USA\\ nikdev@microsoft.com\and S. Matthew Weinberg\thanks{Work done in part while the author was a Research Fellow at the Simons Institute for the Theory of Computing.}\\Princeton University, USA\\ smweinberg@princeton.edu. }
\addtocounter{page}{-1}
\maketitle
\begin{abstract} 

We provide a unified view of many recent developments in Bayesian mechanism design, including the black-box reductions of Cai et al.~\cite{CaiDW13b}, simple auctions for additive buyers~\cite{HartN12}, and posted-price mechanisms for unit-demand buyers~\cite{ChawlaHK07}. Additionally, we show that viewing these three previously disjoint lines of work through the same lens leads to new developments as well. First, we provide a duality framework for Bayesian mechanism design, which naturally accommodates multiple agents and arbitrary objectives/feasibility constraints. Using this, we prove that either a posted-price mechanism or the Vickrey-Clarke-Groves auction with per-bidder entry fees achieves a constant-factor of the optimal revenue achievable by a Bayesian Incentive Compatible mechanism whenever buyers are unit-demand or additive, unifying previous breakthroughs of Chawla et al.~\cite{ChawlaHMS10} and Yao~\cite{Yao15}, and improving both approximation ratios (from $30$ to $24$ and $69$ to $8$, respectively). Finally, we show that this view also leads to improved structural characterizations in the Cai et al. framework.\stepcounter{footnote}\footnotetext{An earlier version of this work appeared under the same title in the proceedings of STOC 2016.}
\end{abstract}



\newpage
\input{intro}

\input{prelim}

\input{duality}
\input{myerson}
\input{flow}

\input{single_bidder}
\input{multi_bidder}

\input{general}

\section{Conclusion}
We present a new duality framework for Bayesian mechanism design, and show how to recover and improve the state-of-the-art mechanisms for additive or unit-demand bidders with independent item values. Additionally, our proofs for the single-item, unit-demand, and additive settings are ``unified'' in the sense that we've separated out part of the proof that is (nearly) identical for all three settings (the duality-based upper bound, Section~\ref{sec:flow}), so that the additional work necessary for each result is minimized (Sections~\ref{sec:single} and~\ref{sec:multi}).\footnote{One way to think of this is that if one wishes to exclusively understand the main result from a single one of~\cite{Myerson81, ChawlaMS10, BabaioffILW14, Yao15}, the quickest way to do so is probably still just to read the original papers. However, if one wishes to understand all of these results, a substantial fraction of the proofs overlap via our approach.}

Additionally, our framework provides a principled starting point for future work (as evidenced by the numerous recent follow-ups of Section~\ref{sec:subsequent}). Many of these works themselves explore new areas and present further open questions (such as competition complexity~\cite{EdenFFTW17a, LiuP17}, limited complementarity~\cite{EdenFFTW17b}, two-sided markets~\cite{BrustleCWZ17}, and ``one-and-a-half'' dimensional settings~\cite{FiatGKK16, DevanurW17, DevanurHP17}). 

Our approach also yields insight into any incentives problem that can be captured by an LP with ``incentive constraints'' and ``feasibility constraints.'' Domains such as signaling (e.g.~\cite{Dughmi14}) and contract theory (e.g.~\cite{Carroll15}) are amenable to such LP formulations. Bayesian Persuasion seems an especially enticing domain, as the algorithmic CDW framework has already found application there~\cite{DughmiX16}. 

\section{Acknowledgements}
We would like to thank Costis Daskalakis and Christos Papadimitriou for numerous helpful discussions during the preliminary stage of this work, and Jason Hartline and Rakesh Vohra for helpful discussion and pointers to related works. 

\bibliographystyle{alpha}
\bibliography{Yang}
\end{document}

%% file: intro.tex

\section{Introduction}
In the past several years, we have seen a tremendous advance in the field of Bayesian Mechanism Design, based on ideas and concepts rooted in Theoretical Computer Science (TCS). For instance, due to a line of work initiated by Chawla, Hartline, and Kleinberg~\cite{ChawlaHK07}, we now know that posted-price mechanisms are approximately optimal with respect to the optimal {Bayesian Incentive Compatible\footnote{A mechanism is Bayesian Incentive Compatible (BIC) if it is in every bidder's interest to tell the truth, assuming that all other bidders' reported their true values. A mechanism is Dominant Strategy Incentive Compatible (DSIC) if it is in every bidder's interest to tell the truth \emph{no matter what reports the other bidders make.} }
(BIC)} mechanism whenever buyers are unit-demand,\footnote{A valuation is unit-demand if $v(S) = \max_{i \in S}\{v(\{i\})\}$. A valuation is additive if $v(S) = \sum_{i \in S} v(\{i\})$.} and values are independent\footnote{That is, the random variables $\{v_{ij}\}_{i,j}$ are independent (where $v_{ij}$ denotes bidder $i$'s value for item $j$).}~\cite{ChawlaHMS10,ChawlaMS10,KleinbergW12}. Due to a line of work initiated by Hart and Nisan \cite{HartN12}, we now know that either running Myerson's auction separately for each item or running the VCG mechanism with a per-bidder entry fee\footnote{By this, we mean that the mechanism offers each bidder the option to participate for $b_i$, which might depend on the other bidders' bids but not bidder $i$'s. If they choose to participate, then they play in the VCG auction (and pay any additional prices that VCG charges them).} is approximately optimal with respect to the optimal BIC mechanism whenever buyers are additive, and values are independent~\cite{LiY13,BabaioffILW14,Yao15}. Due to a line of work initiated by Cai et al.~\cite{CaiDW12a}, we now know that optimal mechanisms are distributions over virtual welfare maximizers, and have computationally efficient algorithms to find them in quite general settings~\cite{CaiDW12b,CaiDW13a,CaiDW13b,BhalgatGM13,DaskalakisW15,DaskalakisDW15}. The main contribution of this work is a unified approach to these three previously disjoint research directions. 
At a high level, we show how a new interpretation of the Cai-Daskalakis-Weinberg (CDW) framework provides us a duality theory, which 
then allows us to strengthen the characterization results of Cai et al., as well as interpret the benchmarks used in~\cite{ChawlaHK07,ChawlaHMS10, ChawlaMS10, KleinbergW12,HartN12,CaiH13, LiY13,BabaioffILW14} as dual solutions. Surprisingly, we learn that \emph{essentially the same dual solution} yields all the key benchmarks in these works. We show how to extend this dual solution to multi-bidder settings, and analyze the mechanisms developed in~\cite{ChawlaHMS10, Yao15} with respect to the resulting benchmarks. In both cases, our analysis yields improved approximation ratios.



\subsection{Simple vs. Optimal Auction Design}
It is well-known by now that optimal multi-item auctions suffer many properties that are undesirable in practice. For example, with just a single additive buyer and two items, the optimal auction could be randomized~\cite{Tha04, Pavlov11a}. Moreover, there exist instances where the buyer's two values are drawn from a correlated distribution where the optimal revenue achieves infinite revenue while the best deterministic mechanism achieves revenue $\leq 1$~\cite{BriestCKW10, HartN13}. Even when the two item values are drawn independently, the optimal mechanism might offer uncountably many different randomized options for the buyer to choose from~\cite{DaskalakisDT13}. Additionally, revenue-optimal multi-item auctions behave non-monotonically: there exist distributions $F$ and $F^+$, where $F^+$ stochastically dominates $F$, such that the revenue-optimal auction when a single additive buyer's values for two items are drawn from $F \times F$ achieves strictly larger revenue than the revenue-optimal auction when a single additive buyer's values are drawn from $F^+ \times F^+$~\cite{HartR12}. Finally, it is known that revenue-optimal auctions may not be DSIC~\cite{Yao17}, and are also \#P-hard to find~\cite{DaskalakisDT14}.

In light of the aforementioned properties, \emph{simple} mechanisms are often used in lieu of \emph{optimal} mechanisms in practice, and an active line of research coined ``simple versus optimal'' mechanism design~\cite{HartlineR09} aims to rigorously understand when simple mechanisms are appropriate in practice. Still, prior work essentially shows that simple mechanisms are never \emph{exactly} optimal, so the main goal of these works is to understand when simple mechanisms are \emph{approximately} optimal.\footnote{On this front, one should not interpret (say) an $8$-approximation as suggesting that sellers should be happy with $1/8$ of the revenue they could potentially achieve. Rather, these guarantees are meant to be interpreted more qualitatively, and suggest claims like ``If simple auction $A$ guarantees a small constant-factor approximation in the worst-case, but simple auction $B$ does not, maybe it's safer to use auction $A$ in practice.''} Some of the most exciting contributions from TCS to Bayesian mechanism design have come from this direction, and include a line of work initiated by Chawla et al.~\cite{ChawlaHK07} for unit-demand buyers, and Hart and Nisan~\cite{HartN12} for additive buyers. 

In a setting with $m$ heterogeneous items for sale and $n$ \emph{unit-demand} buyers whose values for the items are drawn independently, the state-of-the-art shows that a simple posted-price mechanism\footnote{ A posted price mechanism visits each buyer one at a time and posts a price for each item. The buyer can then select any subset of items and pay the corresponding prices. Observe that such a mechanism is DSIC.} obtains a constant factor of the optimal BIC revenue (the revenue of the optimal BIC mechanism)~\cite{ChawlaHK07,ChawlaHMS10,ChawlaMS10,KleinbergW12}. The main idea behind these works is a multi- to single-dimensional reduction. They consider a related setting where each bidder is split into $m$ separate copies, one for each item, with bidder $i$'s copy $j$ interested only in item $j$. The value distributions are the same as the original multi-dimensional setting. One key ingredient driving these works is that the optimal revenue in the original setting is upper bounded by a small constant times the optimal revenue in the copies setting. 

In a setting with $m$ heterogeneous items for sale and $n$ \emph{additive} buyers whose values for the items are drawn independently, the state-of-the-art result shows that for all inputs, either running Myerson's optimal single-item auction for each item separately or running the VCG auction with a per-bidder entry fee obtains a constant factor of the optimal BIC revenue~\cite{HartN12,LiY13,BabaioffILW14,Yao15}. One main idea behind these works is a ``core-tail decomposition'', that breaks the revenue down into cases where the buyers have either low (the core) or high (the tail) values. 

{Although these two approaches appear different at first,} we are able to show that they in fact arise from \emph{basically the same dual} in our duality theory. 
Essentially, we show that a specific dual solution within our framework gives rise to an upper bound that decomposes into the sum of two terms, one that looks like the the copies benchmark, and one that looks like the core-tail benchmark.
In terms of concrete results, {this new understanding yields improved approximation ratios on both fronts. For additive buyers, we improve the ratio provided by Yao~\cite{Yao15} from $69$ to $8$. For unit-demand buyers, we  improve the approximation ratio provided by Chawla et al.~\cite{ChawlaHMS10} from $30$ to $24$.}

{In addition to these concrete results, our work makes the following conceptual contributions as well.}
First, while the single-buyer core-tail decomposition techniques (first introduced by Li and Yao~\cite{LiY13}) are now becoming standard~\cite{LiY13,BabaioffILW14,RubinsteinW15,BateniDHS15}, they do not generalize naturally to multiple buyers. Yao \cite{Yao15} introduced new techniques in his extension to multi-buyers termed ``$\beta$-adjusted revenue'' and ``$\beta$-exclusive mechanisms,'' which are technically quite involved. 
Our duality-based proof can be viewed as a natural generalization of the core-tail decomposition to multi-buyer settings. Second, we use basically the same analysis for both additive and unit-demand valuations, meaning that our framework provides a unified approach to tackle both settings. Finally, we wish to point out that the key difference between our proofs and those of~\cite{ChawlaHMS10,BabaioffILW14,Yao15} are our duality-based benchmarks: we are able to immediately get more mileage out of these benchmarks while barely needing to develop new approximation techniques. Indeed, the bulk of the work is in properly decomposing our benchmarks into terms that can be approximated using  ideas similar to prior work. All these suggest that our techniques are likely be useful in more general settings (and indeed, they have been: see Section~\ref{sec:related}). 


\subsection{Optimal Multi-Dimensional Mechanism Design}
Another recent contribution of the TCS community is the CDW framework for generic Bayesian mechanism design problems. Here, it is shown that Bayesian mechanism design problems for essentially any objective can be solved with black-box access just to an \emph{algorithm} that optimizes a perturbed version of that same objective. That is, even though the original mechanism design problem involves incentives, the optimal BIC mechanism can be found via black-box queries to an algorithm (where the input is known/given and there are no incentives), but this algorithm optimizes a perturbed objective instead. One aspect of this line of work is computational: we now have computationally efficient algorithms to find the optimal (or approximately optimal) mechanism in numerous settings of interest (including the aforementioned cases of many additive/unit-demand buyers, but significantly more general as well). Another aspect is structural: we now know, for instance, that in all settings that fit into this framework, the revenue-optimal mechanism is a distribution over {\em virtual welfare optimizers}.\footnote{Their reduction applies to objectives beyond revenue, such as makespan. The focus of the present paper is on revenue, so we only focus on the projection of their results onto this setting.}  
A mechanism is a virtual welfare optimizer if it 
pointwise optimizes the {\em virtual welfare} (that is, on every input, it selects an outcome that maximizes the virtual welfare). 
The virtual welfare is given by a {\em virtual valuation/transformation}, which is 
 a mapping from valuations to linear combinations of valuations. 

The structural characterization from previous work roughly ends here: the guaranteed virtual transformations were randomized with no promise  any additional properties beyond their existence (and that they could be found in poly-time). Our contribution to this line of work is to improve the existing structural characterization. Specifically, we show that every instance has a strong dual in the form of $n$ disjoint flows, one for each agent. The nodes in agent $i$'s flow correspond to possible types of this agent,\footnote{Both the CDW framework and our duality theory only apply directly if there are finitely many possible types for each agent.} and non-zero flow from type $t_i(\cdot)$ to $t'_i(\cdot)$ captures that the incentive constraint between $t_i(\cdot)$ and $t'_i(\cdot)$ binds. We show how a flow induces a {virtual transformation}, 
and that the optimal dual gives a virtual valuation function such that:
\begin{enumerate}
	\item This virtual valuation function is deterministic and can be found computationally efficiently.
	\item The optimal mechanism has expected revenue $=$ its expected virtual welfare, and every BIC mechanism has expected revenue $\leq $ its expected virtual welfare.
	\item The optimal mechanism optimizes virtual welfare pointwise (i.e. on every input, the virtual welfare maximizing outcome is selected).\footnote{This could be \emph{randomized}; there is always a deterministic maximizer but in cases where the optimal mechanism is randomized, the virtual transformations are such that there are numerous maximizers, and the optimal mechanism selects one of them from a particular probability distribution.}
\end{enumerate}

Here are a few examples of the benefits of such a characterization (which cannot be deduced from~\cite{CaiDW13b}). First, the promised virtual valuation function certifies the optimality of the optimal mechanism: every BIC mechanism has expected revenue $\leq$ its expected virtual welfare, yet the optimal mechanism maximizes virtual welfare pointwise. Second, by looking at the promised virtual valuation function, we can immediately determine which incentive constraints ``matter.'' Specifically, if the flow corresponding to the promised virtual valuation function sends flow from $t_i(\cdot)$ to $t'_i(\cdot)$, then removing the constraint guaranteeing that $t_i(\cdot)$ prefers to tell the truth rather than report $t'_i(\cdot)$ (e.g. through some form of verification) would increase the optimal achievable revenue.\footnote{This claim is only guaranteed to be true in non-degenerate instances with a unique optimal dual - and exactly results from the fact that relaxing tight constraints in non-degenerate LPs improves the optimal solution.} Such a characterization should prove a valuable analytical tool for multi-item auctions, akin to Myerson's virtual values for single-dimensional settings~\cite{Myerson81}.

\subsection{Related Work}\label{sec:related}
\subsubsection{Duality Frameworks}Recently, strong duality frameworks for a single additive buyer were developed in~\cite{DaskalakisDT13,DaskalakisDT15,DaskalakisDT16,GiannakopoulosK14,Giannakopoulos14a,GiannakopoulosK15}. These frameworks show that the dual problem to revenue optimization for a single additive buyer can be interpreted as an optimal transport/bipartite matching problem. Work of Hartline and Haghpanah also provides an alternative ``path-finding'' duality framework for a single additive or unit-demand buyer, and has a more similar flavor to ours (as flows can be interpreted as distributions over paths)~\cite{HartlineH15}. When they exist, these paths provide a witness that a certain Myerson-type mechanism is optimal, but the paths are not guaranteed to exist in all instances. Also similar is independent work of Carroll, which also makes use of a partial Lagrangian over incentive constraints, again for a single additive buyer~\cite{Carroll16}. In addition to their mathematical beauty, these duality frameworks also serve as tools to prove that mechanisms are optimal. These tools have been successfully applied to provide conditions when pricing only the grand bundle (give the buyer the choice only to buy everything or nothing)~\cite{DaskalakisDT13}, posting a uniform item pricing (post the same price on every item)~\cite{HartlineH15}, or even employing a randomized mechanism~\cite{GiannakopoulosK15} is optimal when selling to a single additive or unit-demand buyer. However, none of these frameworks currently applies in multi-bidder settings, and to date have been unable to yield any approximate optimality results in the (single bidder) settings where they do apply.

We also wish to argue that our duality is perhaps more transparent than existing theories. {For instance, it is easy to interpret dual solutions in our framework as virtual valuation functions, and dual solutions for multiple buyer instances just list a dual for each single buyer. In addition,} we are able to re-derive and improve the breakthrough results of~\cite{ChawlaHK07,ChawlaHMS10, ChawlaMS10,HartN12,LiY13,BabaioffILW14,Yao15} \emph{using essentially the same dual solution}. 
Still, it is not our goal to subsume previous duality theories, and our new theory certainly doesn't. 
For instance, previous frameworks are capable of proving that a mechanism is \emph{exactly} optimal when the input distributions are continuous.
Our theory as-is can only handle distributions with finite support exactly.\footnote{Our theory can still handle continuous distributions arbitrarily well. See Section~\ref{sec:prelim}.}
However, we have demonstrated that there is at least one important domain (simple and approximately optimal mechanisms) where our theory seems to be more applicable.

\subsubsection{Related Techniques}
Techniques similar to ours have appeared in prior works as well. For instance, the idea to use Lagrangian multipliers/LP duality for mechanism design dates back at least to early work of Laffont and Robert for selling a single item to budget-constrained bidders~\cite{LaffontR98}, is discussed extensively for instance in~\cite{MyersonGT,Vohra2011mechanism}, and also used for example in recent works as well~\cite{BhalgatGM13,Vohra12}. It is also apparently informal knowledge among some economists that Myerson's seminal result~\cite{Myerson81} can be proved using some form of LP duality, and some versions of these proofs have been published as well (e.g.~\cite{MalakhovV04}). Still, we include in Section~\ref{sec:myerson} a proof of~\cite{Myerson81} in our framework to serve as a warm-up (and because some elements of the proof are simplified via our approach). 

The idea to use ``paths'' of incentive compatibility constraints to upper bound revenue in single-bidder problems dates back at least to work of Rochet and Chon\'{e} studying revenue optimization in general multi-item settings~\cite{RochetC98}, and Armstrong~\cite{Armstrong96,Armstrong99}, which studies the special case of a single bidder and two items. More recently, Cai et al. use this approach to prove hardness of approximation for a single bidder with submodular valuations for multiple items~\cite{CaiDW13b}, Hartline and Haghpanah provide sufficient conditions for especially simple mechanisms to be optimal for a single unit-demand or additive bidder~\cite{HartlineH15}, and Carroll proves that selling separately is max-min optimal for a single additive buyer when only the marginals are known but not the (possibly correlated) joint value distribution~\cite{Carroll16}. Indeed, many of these works also observe that the term ``virtual welfare'' is appropriate to describe the resulting upper bounds. Still, these works focus exclusively on providing conditions for certain mechanisms to be \emph{exactly optimal}, and therefore impose some technical conditions on the settings where they apply. In comparison, our work pushes the boundaries by accommodating both approximation (in the sense that our framework can prove that simple mechanisms are approximately optimal and not just that optimal mechanisms are optimal) and unrestricted settings (in the sense that our framework isn't restricted to a single buyer, or additive/unit-demand valuations).


Finally, we note that some of the benchmarks used in later sections can be derived without appealing to duality~\cite{ChawlaMS10}. Therefore, a duality theory is not ``necessary'' in order to obtain our benchmarks. Still, prior to our work it was unknown that these benchmarks were at all useful outside of the unit-demand settings for which they were developed. Additionally, both the primal and dual understanding of these benchmarks is valuable for extending the state-of-the-art, discussed in more detail below. Prior work has also obtained approximately optimal auctions via some sort of ``benchmark decomposition'' in the unrelated setting of digital goods~\cite{ChenGL15}. 

\subsubsection{Approximation in Multi-Dimensional Mechanism Design}
Finally, we provide a brief overview of recent work providing simple and approximately optimal mechanisms in multi-item settings. Seminal work of Chawla, Hartline, and Kleinberg proves that a posted-price mechanism gets a 3-approximation to the optimal deterministic mechanism for a single unit-demand buyer with independently drawn item values~\cite{ChawlaHK07}. Chawla et al. improve the ratio to 2, and prove a bound of 6.75 against the optimal deterministic, DSIC mechanism for multiple buyers~\cite{ChawlaHMS10}. Chawla, Malec, and Sivan show that the bound degrades by at most a factor of 5 when comparing to the optimal randomized, BIC mechanism~\cite{ChawlaMS10}. Kleinberg and Weinberg improve the bound of 6.75 to 6~\cite{KleinbergW12}. Roughgarden, Talgam-Cohen and Yan provide a prior-independent ``supply-limiting'' mechanism in this setting, and also prove a Bulow-Klemperer~\cite{BulowK96} result: the VCG mechanism with additional bidders yields more expected revenue than the optimal deterministic, DSIC mechanism (with the original number of bidders), when bidder valuations are unit-demand, i.i.d., and values for items are regular and independent (possibly asymmetric)~\cite{RoughgardenTY12}. All of these results get mileage from the ``\copies'' benchmark initiated in~\cite{ChawlaHK07}.  

More recent influential work of Hart and Nisan proves that selling each item separately at its Myerson reserve\footnote{The Myerson reserve of a one-dimensional distribution refers to the revenue-optimal price to set if one seller is selling only this item to a single buyer.} gets an $O(\log^2 m)$ approximation to the optimal mechanism for a single additive buyer and $m$ independent (possibly asymmetric) items~\cite{HartN12}. Li and Yao improve this to $O(\log m)$, which is tight~\cite{LiY13}. Babaioff et al. prove that the better of selling separately and bundling together gets a 6-approximation~\cite{BabaioffILW14}. Bateni et al. extend this to a model of limited correlation~\cite{BateniDHS15}. Rubinstein and Weinberg extend this to a single buyer with ``subadditive valuations over independent items''~\cite{RubinsteinW15}. Yao shows that the better of selling each item separately using Myerson's auction and running VCG with a per-bidder entry fee gets a $69$-approximation when there are many additive buyers and all values for all items are independent~\cite{Yao15}. Goldner and Karlin show how to use these results to obtain approximately optimal prior-independent mechanisms for many additive buyers~\cite{GoldnerK16}. More recently, Chawla and Miller show that a posted-price mechanism with per-bidder entry fee gets a constant-factor approximation for many bidders with ``additive valuations subject to matroid constraints''~\cite{ChawlaM16}. All of these results get mileage from the ``core-tail decomposition'' initiated in~\cite{LiY13}. 

The present paper unifies these two lines of work by showing that the \copies\ benchmark and the core-tail decomposition both arise from essentially the same dual in our duality theory. These benchmarks provide a necessary starting point for the above results, but proving guarantees against these benchmarks of course still requires significant work. We believe that our duality theory now provides the necessary starting point to extend these results to much more general settings, as evidenced by the follow-up works discussed below.

\subsubsection{Subsequent Work}\label{sec:subsequent} Since the presentation of an earlier version of this work at STOC 2016, numerous follow-up works have successfully made use of our framework to design (approximately) optimal auctions in much more general settings. For example, Cai and Zhao show that the better of a posted-price mechanism and an anonymous posted-price mechanism with per-bidder entry fee gets a constant-factor approximation for many bidders with ``\emph{XOS valuations over independent items}''~\cite{CaiZ17}. This extends the previous state-of-the-art~\cite{ChawlaM16} from Gross Substitutes to XOS valuations. Eden et al. show that the better of selling separately and bundling together gets an $O(d)$-approximation for a single bidder with ``complementarity-$d$ valuations over independent items''~\cite{EdenFFTW17a}. The same authors also prove a Bulow-Klemperer result: the VCG mechanism with additional bidders yields more expected revenue than the optimal randomized, BIC mechanism (with the original number of bidders), when bidder valuations are ``additive subject to downward closed constraints,'' i.i.d., and values for items are regular and independent (possibly asymmetric)~\cite{EdenFFTW17b}. Brustle et al. design a simple mechanism that achieves  $\frac{1}{2}$ of the optimal gains from trade in certain \emph{two-sided markets}, such as bilateral trading and double auctions~\cite{BrustleCWZ17}. Devanur and Weinberg provide an alternative proof of Fiat et al.'s solution to the ``FedEx Problem,'' and extend it to design the optimal auction for a single buyer with a private budget~\cite{FiatGKK16, DevanurW17}. Finally, Liu and Psomas provide a Bulow-Klemperer result for \emph{dynamic auctions}~\cite{LiuP17}, and Fu et al. design approximately optimal BIC mechanisms for correlated bidders~\cite{FuLLT17}.

\vspace{.1in}
\noindent\textbf{Organization.} We provide preliminaries and notation below. In Section~\ref{sec:duality}, we present our duality theory for revenue maximization in the special case of additive/unit-demand bidders. In Section~\ref{sec:myerson}, we present a duality proof of Myerson's seminal result, and in Section~\ref{sec:flow} we present a canonical dual solution that proves useful in different settings. As a warm-up, we show in Section~\ref{sec:single} how to analyze this dual solution when there is just a single buyer. In Section~\ref{sec:multi}, we provide the multi-bidder analysis, which is more technical. In Section~\ref{sec:general}, we conclude with a formal statement of our duality theory in general settings.

%% file: prelim.tex

\section{Preliminaries}\label{sec:prelim}

\noindent\textbf{Optimal Auction Design.} For the bulk of the paper, we will study the following setting (in Section~\ref{sec:general}, we will show that our duality theory holds much more generally). The buyers (we will use the terms buyer and bidder interchangeably) are either all unit-demand or all additive, with buyer $i$ having value $t_{ij}$ for item $j$. Recall that a valuation is unit-demand if $v(S) = \max_{i \in S}\{v(\{i\})\}$ and a valuation is additive if $v(S) = \sum_{i \in S} v(\{i\})$. {We use $t_{i}=(t_{i1},\ldots, t_{im})$ to denote buyer $i$'s values for all the goods and $t_{-i}$ to denote every buyer except $i$'s values for all the goods. $T_{ij}$ is the set of all possible values of buyer $i$ for item $j$, $T_{i}=\times_{j} T_{ij}$, $T_{-i} = \times_{i^{*}\neq i} T_{i^{*}}$ and $T = \times_{i} T_{i}$.}
All values for all items are drawn independently. We denote by $D_{ij}$ the distribution of $t_{ij}$, $D_i = \times_{j} D_{ij}$, $D_{i,-j}=\times_{j^{*}\neq j} D_{ij^{*}}, $  $D = \times_i D_i$, and $D_{-i} = \times_{i^* \neq i} D_{i^*}$, and $f_{ij}$ (or $f_i, f_{i,-j}, f_{-i}$, etc.) the densities of these finite-support distributions (that is, $f_i(x) = \Prob_{t_i \sim D_i}[t_i = x]$). We define $\mathcal{F}$ to be a set system over $[n]\times[m]$ that describes all feasible allocations.\footnote{When bidders are additive, $\mathcal{F}$ only allows allocating each item at most once. When bidders are unit-demand, $\mathcal{F}$ contains all matchings between the bidders and the items.}

A \emph{mechanism} takes as input a reported type from each bidder and selects (possibly randomly) an outcome in $\mathcal{F}$, and payments to charge the bidders. A mechanism is Bayesian Incentive Compatible (BIC) if it is in each buyers interest to report their true type, assuming that the other buyers do so as well, and Bayesian Individually Rational (BIR) if each buyer gets non-negative utility for reporting their true type (assuming that the other bidders do so as well). The revenue of an auction is simply the expected sum of payments made when bidders drawn from $D$ report their true values. The \emph{optimal auction} optimizes expected revenue over all BIC and BIR mechanisms. For a given value distribution $D$, we denote by $\rev(D)$ the expected revenue achieved by this auction, and it will be clear from context whether buyers are additive or unit-demand. For a specific BIC mechanism $M$, we will also use $\rev^M(D)$ to denote the expected revenue achieved by $M$ when bidders with valuations drawn from $D$ report truthfully.

\vspace{.1in}

\noindent\textbf{Reduced Forms.} The reduced form of an auction stores for all bidders $i$, items $j$, and types $t_i$, the probability that bidder $i$ will receive item $j$ when reporting $t_i$ to the mechanism (over the randomness in the mechanism and randomness in other bidders' reported types, assuming they come from $D_{-i}$) as $\pi_{ij}(t_i)$. It is easy to see that if a buyer is additive, or unit-demand and receives only one item at a time, that their expected value for reporting type $t'_i$ to the mechanism is just $t_i \cdot \pi_i(t'_i)$ (where we treat $t_i$ and $\pi_i(t'_i)$ as vectors, and $\cdot$ denotes a vector dot-product). We say that a reduced form is \emph{feasible} if there exists some feasible mechanism (that ex-post selects an outcome in $\mathcal{F}$ with probability $1$) that matches the probabilities promised by the reduced form. If $\polytope$ is defined to be the set of all feasible reduced forms, it is easy to see (and shown in~\cite{CaiDW12a}, for instance) that $\polytope$ is closed and convex.

We will also use $p_i(t_i)$ to refer to the expected payment made by bidder $i$ when reporting $t_i$ to the mechanism (over the randomness in the mechanism and randomness in other bidders' reported types, assuming they come from $D_{-i}$). 
\vspace{.1in}

\noindent\textbf{Simple Mechanisms.} Even though the benchmark we target is the optimal \emph{randomized} BIC mechanism, the simple mechanisms we design will all be deterministic and satisfy DSIC. For a single buyer, the two mechanisms we consider are selling separately and selling together. Selling separately posts a price $p_j$ on each item $j$ and lets the buyer purchase whatever subset of items she pleases. We denote by $\srev(D)$ the revenue of the optimal such pricing. Selling together posts a single price $p$ on the grand bundle, and lets the buyer purchase the entire bundle for $p$ or nothing. We denote by $\brev(D)$ the revenue of the optimal such pricing. For multiple buyers the generalization of selling together is the VCG mechanism with an entry fee, which offers to each bidder $i$ the opportunity to pay an entry fee $e_i(t_{-i})$ and participate in the VCG mechanism (paying any additional fees charged by the VCG mechanism). If they choose not to pay the entry fee, they pay nothing and receive nothing. We denote the revenue of the mechanism that charges the optimal entry fees to the buyers as $\bvcg(D)$, and $\vcg(D)$ the revenue of the VCG mechanism with no entry fees. The generalization of selling separately is a little different, and described immediately below.

\vspace{.1in}
\noindent\textbf{Single-Dimensional Copies.} A benchmark that shows up in our decompositions relates the multi-dimensional instances we care about to a single-dimensional setting, and originated in work of Chawla et. al.~\cite{ChawlaHK07}. For any multi-dimensional instance $D$ we can imagine splitting bidder $i$ into $m$ different copies, with bidder $i$'s copy $j$ interested only in receiving item $j$ and nothing else. So in this new instance there are $nm$ single-dimensional bidders, and copy $(i,j)$'s value for winning is $t_{ij}$ (which is still drawn from $D_{ij}$). 
The set system $\mathcal{F}$ from the original setting now specifies which copies can  simultaneously win. We denote by \copies$(D)$ the revenue of Myerson's optimal auction~\cite{Myerson81} in the copies setting induced by $D$.\footnote{Note that when buyers are additive that \copies\ is exactly the revenue of selling items separately using Myerson's optimal auction in the original setting.}

\vspace{.1in}
\noindent\textbf{Continuous versus Finite-Support Distributions.} Our approach explicitly assumes that the input distributions have finite support. This is a standard assumption when computation is involved. However, most existing works in the simple vs. optimal paradigm hold even for continuous distributions (including~\cite{ChawlaHK07,ChawlaHMS10,ChawlaMS10,HartN12,LiY13,BabaioffILW14,Yao15,RubinsteinW15,BateniDHS15}). Fortunately, it is known that every $D$ can be discretized into $D^+$ such that $\rev(D) \in [(1-\epsilon)\rev(D^+),(1+\epsilon)\rev(D^+)]$, and $D^+$ has finite support. So all of our results can be made arbitrarily close to exact for continuous distributions. We conclude this section by proving this formally, making use of the following theorem proved in~\cite{RubinsteinW15}, which draws from prior works~\cite{HartlineL10,HartlineKM11,BeiH11,DaskalakisW12}. Note that the theorem below holds for distributions over arbitrary valuation functions $t_i(\cdot)$, and not just additve/unit-demand.

\begin{theorem}\cite{RubinsteinW15,DaskalakisW12}\label{thm:RW15}
Let $M$ be any BIC mechanism for values drawn from distribution $D$, and for all $i$, let $D_i$ and $D_i^+$ be any two distributions, with coupled samples $t_i(\cdot)$ and $t_i^+(\cdot)$ such that $t_i^+(x) \geq t_i(x)$ for all $x \in \mathcal{F}$. If $\delta_i(\cdot) = t_i^+(\cdot) - t_i(\cdot)$, then for any $\epsilon > 0$, there exists a BIC mechanism $M'$ such that {$\rev^{M'}(D^+) \geq (1-\epsilon)(\rev^M(D) - \frac{\textsc{Val}(\delta)}{\epsilon})$}, where $\textsc{Val}(\delta)$ denotes the expected welfare of the VCG allocation when buyer $i$'s type is drawn according to the random variable $\delta_i(\cdot)$. 
\end{theorem}

To see how this implies that our duality is arbitrarily close to exact for continuous distributions, let $D^\epsilon_i$ be the distribution that first samples $t_i(\cdot)$ from $D_i$, then outputs $t^\epsilon_i(\cdot)$ such that $t^\epsilon_i(x) = t_i(x) \cdot \I(t_i([m]) \leq 1/\epsilon)$.\footnote{That is, if the value of $t_i$ for the grand bundle satisfies $t_i([m]) \leq 1/\epsilon$, then $t^\epsilon_i(x) = t_i(x)$. Otherwise, $t_i^\epsilon(x) = 0$ for all $x$.} It is easy to see that as $\epsilon \rightarrow 0$, $\rev^M(D^\epsilon) \rightarrow \rev^M(D)$: for every mechanism $M$ and every $\eta > 0$, there exists an $\epsilon > 0$ such that a $(1-\eta)$ fraction of $M$'s revenue when buyers' types are drawn from $D$ comes from buyers with $t_i([m])\leq 1/\epsilon$\footnote{Similarly, if $M$ achieves infinite revenue, then for every $\eta > 0$, there exists an $\epsilon > 0$ such that the revenue of $M$ when buyers' types are drawn from $D$ coming from buyers with $t_i([m]) \leq 1/\epsilon$ is at least $1/\eta$. So our approach will still show that whenever $\rev(D)$ is infinite, the revenue of the approximately optimal mechanisms we use is unbounded.}. For the chosen $\epsilon$,  $M$ is still a BIC mechanism when buyers' types are drawn from $D^{\epsilon}$, and its revenue under $D^{\epsilon}$ is at least $(1-\eta)$ fraction of its revenue under $D$.  So we can get arbitrarily close while only considering distributions that are bounded.

Now for any bounded distribution $D_i$, define $D^{+,\epsilon}_i$ to first sample $t_i(\cdot)$ from $D_i$, then output $t^{+,\epsilon}_i(\cdot)$ such that $t^{+,\epsilon}_i(x) = \epsilon^2 \cdot \lceil t_i(x)/\epsilon^2 \rceil$. Similarly define $D^{-,\epsilon}_i$ to first sample $t_i(\cdot)$ from $D_i$, then output $t^{-,\epsilon}_i(\cdot)$ such that $t^{-,\epsilon}_i(x) = \epsilon^2 \cdot \left(\lceil t_i(x)/\epsilon^2 \rceil -1 \right)$. Then it's clear that $D^{+,\epsilon}_i, D_i,$ and $D^{-,\epsilon}_i$ can be coupled so that $t^{+,\epsilon}_i(x) \geq t_i(x) \geq t^{-,\epsilon}_i(x)$ for all $x$, and that taking either of the two consecutive differences results in a $\delta_i(\cdot)$ such that $\delta_i(x) \leq \epsilon^2$ for all $x$. So for any desired $\epsilon$, applying Theorem~\ref{thm:RW15} with $M$ as the optimal mechanism for $D$, we get a mechanism $M'$ for $D^{+,\epsilon}$ with revenue at least $(1-\epsilon)\rev(D)-n\epsilon$. Similarly, applying Theorem~\ref{thm:RW15} with $M$ as the optimal mechanism for $D^{-,\epsilon}$, we get a mechanism $M'$ for $D$ with revenue at least $(1-\epsilon)\rev(D^{-,\epsilon}) - n\epsilon$. Together, these claims imply that $\rev(D) \in [(1-\epsilon)\rev(D^{-,\epsilon})-n\epsilon, \frac{\rev(D^{+,\epsilon})}{1-\epsilon} +\frac{n\epsilon}{1-\epsilon}]$.  Finally, we just observe that $\rev(D^{+,\epsilon}) = \rev(D^{-,\epsilon}) + n\epsilon$ (the revenue of the optimal mechanism increases by exactly $n\epsilon$ going from $D^{-,\epsilon}$ to $D^{+,\epsilon}$), as every buyer values every outcome at exactly $\epsilon$ more in $D^{+,\epsilon}$ versus $D^{-,\epsilon}$. So as $\epsilon \rightarrow 0$, both approach $\rev(D)$. Note that both $D^{+,\epsilon}$ and $D^{-,\epsilon}$ have finite support, so our theory will directly design $M$ that achieve constant-factor approximations for $\rev(D^{-,\epsilon})$.

%% file: duality.tex

\section{Our Duality Theory}\label{sec:duality}
In this section we provide our duality framework, specialized to unit-demand/additive bidders. We begin by writing the linear program (LP) for revenue maximization (\Cref{fig:LPRevenue}). 
 For ease of notation, assume that there is a special type $\varnothing$ to represent the option of not participating in the auction. That means ${\pi}_i(\varnothing)=\textbf{0}$ and $p_{i}(\varnothing)=0$. Now a Bayesian IR (BIR) constraint is simply another BIC constraint: for any type $t_{i}$, bidder $i$ will not want to lie to type $\varnothing$.  We let $T_{i}^{+}=T_{i}\cup \{\varnothing\}$. 
 To proceed, we will introduce a variable $\lambda_i(t,t')$ for each of the BIC constraints,
and take the partial Lagrangian of LP~\ref{fig:LPRevenue} by Lagrangifying all BIC constraints. 
The theory of Lagrangian multipliers tells us that the solution to  LP~\ref{fig:LPRevenue} is equivalent to the primal variables solving the partially Lagrangified dual (\Cref{fig:Lagrangian}). 

\begin{figure}[ht]
\colorbox{MyGray}{
\begin{minipage}{0.98\textwidth} {
\noindent\textbf{Variables:}
\begin{itemize}
\item $p_i(t_i)$, for all bidders $i$ and types $t_i \in T_i$, denoting the expected price paid by bidder $i$ when reporting type $t_i$ over the randomness of the mechanism and the other bidders' types.
\item $\pi_{ij}(t_i)$, for all bidders $i$, items $j$, and types $t_i \in T_i$, denoting the probability that bidder $i$ receives item $j$ when reporting type $t_i$ over the randomness of the mechanism and the other bidders' types.
\end{itemize}
\textbf{Constraints:}
\begin{itemize}
\item ${\pi}_i(t_i) \cdot t_i - p_i(t_i) \geq {\pi}_i(t'_i)\cdot t_i - p_i(t'_i) $, for all bidders $i$, and types $t_i \in T_i, t'_i \in T_i^+$, guaranteeing that the reduced form mechanism $({\pi},{p})$ is BIC and BIR.
\item ${\pi} \in \polytope$, guaranteeing ${\pi}$ is feasible.
\end{itemize}
\textbf{Objective:}
\begin{itemize}
\item $\text{Maximize:} \sum_{i=1}^{n} \sum_{t_i \in T_i} f_{i}(t_{i})\cdot p_i(t_i)$, the expected revenue.\\
\end{itemize}}
\end{minipage}}
\caption{A Linear Program (LP) for Revenue Optimization.}
\label{fig:LPRevenue}
\end{figure}

\begin{figure}[ht]
\colorbox{MyGray}{
\begin{minipage}{0.98\textwidth} {
\noindent\textbf{Variables:}
\begin{itemize}
\item $\lambda_i(t_{i},t_{i}')$ for all $i,t_{i}\in T_{i},t_{i}' \in T_i^{+}$, the Lagrangian multipliers for Bayesian IC constraints.
\end{itemize}
\textbf{Constraints:}
\begin{itemize}
\item $\lambda_i(t_{i},t_{i}')\geq 0$ for all $i,t_{i}\in T_{i},t_{i}' \in T_i^{+}$, guaranteeing that the Lagrangian multipliers are non-negative. 
\end{itemize}
\textbf{Objective:}
\begin{itemize}
\item $\text{Minimize:} \max_{\pi\in\polytope, p} \L(\lambda, \pi, p)$.\\
\end{itemize}}
\end{minipage}}
\caption{Partial Lagrangian of the Revenue Maximization LP.}
\label{fig:Lagrangian}
\end{figure}
\begin{definition}
Let $\L(\lambda, \pi, p)$ be the partial Lagrangian defined as follows:
\begin{align*}
 \L(\lambda, \pi, p)=\sum_{i=1}^{n} \left(\sum_{t_i \in T_i} f_{i}(t_{i})\cdot p_i(t_i)+\sum_{t_{i}\in T_{i}}\sum_{t_{i}'\in T_i^{+}} \lambda_{i}(t_{i},t_{i}')\cdot \Big(t_{i}\cdot\big(\pi(t_{i})-\pi({t_{i}'})\big)-\big(p_{i}(t_{i})-p_{i}(t_{i}')\big)\Big)\right)\stepcounter{equation}\tag{\theequation} \label{eq:primal lagrangian}
\end{align*}
\begin{align*}
~=\sum_{i=1}^n \sum_{t_{i}\in T_{i}} p_{i}(t_{i})\Big(f_{i}(t_{i})+&\sum_{t_{i}'\in T_{i}} \lambda_{i}(t_{i}',t_{i})-\sum_{t_{i}'\in T_{i}^{+}} \lambda_{i}(t_{i},t_{i}')\Big)\\
&+\sum_{i=1}^n\sum_{t_{i}\in T_{i}}\pi_{i}(t_{i}) \Big(\sum_{t_{i}'\in T_{i}^{+}}t_i \cdot \lambda_{i}(t_{i},t_{i}')-\sum_{t'_{i}\in T_{i}}t_{i}'\cdot \lambda_{i}(t_{i}',t_{i})\Big)\stepcounter{equation}\tag{\theequation} \label{eq:dual lagrangian}
\end{align*}
\end{definition}
\notshow{
({\pi}_i(\varnothing)=\textbf{0},\ p_{i}(\varnothing)=0)

=& \sum_{i=1}^{n}\Big(\sum_{t_{i}\in T_{i}} p_{i}(t_{i})\cdot\big(f_{i}(t_{i})+\sum_{t_{i}'\in T_{i}} \lambda_{i}(t_{i}',t_{i})-\sum_{t_{i}'\in T_{i}^{+}} \lambda_{i}(t_{i},t_{i}')\big)\\
&~~+\sum_{t_{i}\in T_{i}}\pi_{i}(t_{i})\cdot \big(t_{i}\cdot \sum_{t_{i}'\in T_{i}^{+}}\lambda_{i}(t_{i},t_{i}')-\sum_{t'_{i}\in T_{i}}(t_{i}'\cdot \lambda_{i}(t_{i}',t_{i}))\big)\Big)\\ &~~~~~~~~~~~~({\pi}_i(\varnothing)=\textbf{0},\ p_{i}(\varnothing)=0)\stepcounter{equation}\tag{\theequation} \label{eq:dual lagrangian}}

\subsection{Useful Properties of the Dual Problem}
In this section, we make some observations about the dual problem to get some traction on what duals might induce useful upper bounds. 
\begin{definition}[Useful Dual]
A feasible dual solution $\lambda$ is \textbf{useful} if $\max_{\pi\in\polytope, p} \L(\lambda, \pi, p)< \infty$.
\end{definition}

\begin{lemma}[Useful Dual]\label{lem:useful dual}
A dual solution $\lambda$ is useful if and only if for each bidder $i$, $\lambda_{i}$ forms a valid flow, i.e., iff the following satisfies flow conservation (flow in = flow out) at all nodes except the source and the sink:
\begin{itemize}
\item Nodes: A super source $s$ and a super sink $\varnothing$, along with a node $t_{i}$ for every type $t_{i}\in T_{i}$. 
\item Flow from $s$ to $t_{i}$ of weight $f_i(t_{i})$, for all $t_{i}\in T_{i}$.
\item Flow from $t$ to $t'$ of weight $\lambda_i(t,t')$ for all $t\in T$, and $t'\in T_{i}^{+}$ (including the sink $\varnothing$).
\end{itemize}
\end{lemma}

\begin{proof}
Let us think of $\L(\lambda, \pi, p)$ using expression~(\ref{eq:dual lagrangian}). Clearly, if there exists any $i$ and $t_{i}\in T_{i}$ such that $$f_{i}(t_{i})+\sum_{t_{i}'\in T_{i}} \lambda_i(t_{i}',t_{i})-\sum_{t_{i}'\in T_{i}^{+}} \lambda_i(t_{i},t_{i}')\neq 0,$$ then since $p_{i}(t_{i})$ is unconstrained (note that we do \textbf{not} include the constraints $p_i(t_i) \geq 0$ in Figure~\ref{fig:LPRevenue}, so these variables are indeed unconstrained) and has a non-zero multiplier in the objective, $\max_{\pi\in\polytope, p} \L(\lambda, \pi, p)=+\infty$. Therefore, in order for $\lambda$ to be useful, we must have $$f_{i}(t_{i})+\sum_{t_{i}'\in T_{i}} \lambda_i(t_{i}',t_{i})-\sum_{t_{i}'\in T_{i}^{+}} \lambda_i(t_{i},t_{i}')=0$$ for all $i$ and $t_{i}\in T_{i}$. This is exactly saying what we described in the Lemma statement is a flow. The other direction is simple, whenever $\lambda$ forms a flow, $\L(\lambda, \pi, p)$ only depends on $\pi$. Since $\pi$ is bounded, the maximization problem has a finite value.
\end{proof}

\begin{definition}[Virtual Value Function]\label{def:virtual value}
For each $\lambda$, we define a corresponding virtual value function $\Phi^\lambda(\cdot)$, such that for every bidder $i$, every type $t_{i}\in T_{i}$, $\Phi^\lambda_{i}(t_{i})=t_{i}-{1\over f_{i}(t_{i})}\sum_{t_{i}'\in T_{i}} \lambda_{i}(t_{i}',t_{i})(t_{i}'-t_{i}).$ Note that for all $i$, $\Phi^\lambda_i(\cdot)$ is a vector-valued function, so we use $\Phi^\lambda_{ij}(\cdot)$ to refer to the $j^{th}$ component of $\Phi^\lambda_i(\cdot)$, and refer to this as bidder $i$'s virtual value for item $j$.
\end{definition}

\begin{theorem}[Virtual Welfare $\geq$ Revenue]\label{thm:revenue less than virtual welfare}
Let $\lambda$ be any useful dual solution and $M = (\pi,p)$ be any BIC mechanism. The revenue of $M$ is less than or equal to the virtual welfare of $\pi$ w.r.t. the virtual value function $\Phi^\lambda(\cdot)$ corresponding to $\lambda$. That is:
$$\sum_{i=1}^{n} \sum_{t_i \in T_i} f_{i}(t_{i})\cdot p_i(t_i)\leq \sum_{i=1}^{n} \sum_{t_{i}\in T_{i}} f_{i}(t_{i})\cdot \pi_{i}(t_{i})\cdot\Phi^\lambda_{i}(t_{i}).$$
Equality holds if and only if for all $i, t, t'$ such that $\lambda_i(t, t') > 0$, the BIC constraint for bidder $i$ between $t$ and $t'$ binds in $M$ (that is, bidder $i$ with type $t$ is indifferent between reporting $t$ and $t'$). Furthermore, let $\lambda^{*}$ be the optimal dual variables and $M^{*}=(\pi^{*},p^{*})$ be the revenue-optimal BIC mechanism, then the expected virtual welfare with respect to $\Phi^{*}$ (induced by $\lambda^{*}$) under  $\pi^{*}$ equals the expected revenue of $M^{*}${, and $$\pi^* \in \argmax_{\pi \in \polytope}\left\{\sum_{i=1}^n \sum_{t_i \in T_i} f_i(t_i) \cdot \pi_i(t_i) \cdot \Phi^*_i(t_i)\right\}.$$}
\end{theorem}

\begin{proof}
When $\lambda$ is useful, we can simplify $\L(\lambda, \pi, p)$ by removing all terms associated with $p$ (because all such terms have a multiplier of zero, by Lemma~\ref{lem:useful dual}), and replace the terms $\sum_{t_{i}'\in T_{i}^{+}}\lambda_i(t_{i},t_{i}')$ with $f_{i}(t_{i})+\sum_{t_{i}'\in T_{i}} \lambda_i(t_{i}',t_{i})$. After the simplification, we have $\L(\lambda, \pi, p) = \sum_{i=1}^{n} \sum_{t_{i}\in T_{i}} f_{i}(t_{i})\cdot \pi_{i}(t_{i})\cdot\Big(t_{i}-{1\over f_{i}(t_{i})}\cdot$ $\sum_{t_{i}'\in T_{i}} \lambda_{i}(t_{i}',t_{i})(t_{i}'-t_{i})\Big)$, which  equals $\sum_{i=1}^{n} \sum_{t_{i}\in T_{i}} f_{i}(t_{i})\cdot \pi_{i}(t_{i})\cdot\Phi^\lambda_{i}(t_{i})$, exactly the virtual welfare of $\pi$. Now, we only need to prove that $\L(\lambda, \pi, p)$ is greater than the revenue of $M$. Let us think of $\L(\lambda, \pi, p)$ using Expression~\eqref{eq:primal lagrangian}. Since $M$ is a BIC mechanism, $t_i \cdot \big(\pi(t_{i})-p_i({t_{i}'})\big)-\big(p_{i}(t_{i})-p_{i}(t_{i}')\big)\geq 0$ for any $i$ and $t_{i}\in T_{i}$, $t'_i \in T^+_i$. Also, all the dual variables $\lambda$ are nonnegative. Therefore, it is clear that $\L(\lambda, \pi, p)$ is at least as large as the revenue of $M$. Moreover, if the BIC constraint for bidder $i$ between $t$ and $t'$ binds in $M$ for all $i, t, t'$ such that $\lambda_i(t, t') > 0$, then we in fact have $\L(\lambda, \pi, p) = \sum_i \sum_{t_i} f_i(t_i) p_i(t_i)$, so the revenue of $M$ is equal to its expected virtual welfare under $\Phi^\lambda(\cdot)$ (because the Lagrangian terms added to the revenue are all zero). 

When $\lambda^{*}$ is the optimal dual solution, by {strong LP duality} applied to the LP of Figure~\ref{fig:LPRevenue}, we know $\max_{\pi\in\polytope, p}\L(\lambda^{*}, \pi, p)$ equals the revenue of $M^{*}$. But we also know that $\L(\lambda^{*}, \pi^{*}, p^{*})$ is at least as large as the revenue of $M^{*}$, so $\pi^{*}$ necessarily maximizes the virtual welfare over all $\pi \in \polytope$, with respect to the virtual transformation $\Phi^*$ corresponding to $\lambda^*$.
\end{proof}

To summarize: we have shown that every flow induces a finite upper bound on how much revenue a BIC mechanism can possibly achieve. We have also observed that this upper bound can be interpreted as the maximum virtual welfare obtainable with respect to a virtual valuation function that is decided by the flow. In the next two sections, we will instantiate this theory by designing specific flows, and obtain benchmarks that upper bound the optimal revenue. 

%% file: myerson.tex

\section{Canonical Flow for a Single Item}\label{sec:myerson}
In this section, we provide a canonical flow for single-item settings, and show that it implies the main result from Myerson's seminal work~\cite{Myerson81}. Essentially, Myerson proposes a specific virtual valuation function and shows that for this virtual valuation function, the expected revenue of any BIC mechanism is always upper bounded by its expected virtual welfare. Moreover, he describes an ironing procedure to guarantee that this virtual valuation is monotone, and proves that the revenue-optimal mechanism simply awards the item to the bidder with the highest virtual value. Myerson's proof is quite elegant, and we are not claiming that our proof below is simpler.\footnote{If one's goal is simply to understand Myerson's result and nothing more, the original proof and ours are comparable in simplicity. Many alternative comparably simple proofs exist as well, some of which are not much different than ours (e.g.~\cite{MalakhovV04}).} The purpose of the proof below is:
\begin{itemize}
\item Serve as a warm-up for the reader to get comfortable with flows and virtual valuations.
\item Separate out parts of the proof that can be directly applied to more general settings (e.g. Theorem~\ref{thm:revenue less than virtual welfare}). 
\item Provide a specific flow that will be used in later sections to provide benchmarks in multi-item settings.
\end{itemize}

In this section, we will have $m = 1$ and drop the item subscript $j$. We begin with a definition of Myerson's (ironed) virtual valuation function, adapted to the discrete setting. 

\begin{definition}[Single-dimensional Virtual Value]\label{def:vv}
For a single-dimensional discrete discribution $D_i$, if $f_i(t_i) > 0$, let $\varphi^{D_i}_i(t_i) = t_i - \frac{(t'_i-t_i)\cdot \Prob_{t \sim D_i}[t>t_i]}{f_i(t_i)}$, where $t'_i = \min_{t > t_i, t \in T_i}\{t\}$. If $f_i(t_i) = 0$, let $\varphi^{D_i}_i(t_i) = 0$.\footnote{Actually we could define $\varphi^{D_i}_i(t_i)$ arbitrarily and everything that follows will still hold.} If the distribution $D_i$ is clear from context, we will just write $\varphi_i(t_i)$. 
\end{definition}

The ironing procedure described below essentially finds any non-monotonicities in $\varphi_i(\cdot)$ and ``irons'' them out. Note that Steps 4 and 5 maintain that ironed virtual values are consistent within any ironed interval.

\begin{definition}[Ironing]\label{def:ironing} Let $\sim$ be an equivalence relation on the support of $D_i$, and $\tilde{\varphi}(\cdot)$ be the \emph{ironed virtual valuation} function defined in the following way. We say that an interval $[t^*_i,t_i]$ is \emph{ironed} if $t \sim t'$ for all $t, t' \in [t^*_i, t_i]$. 
\begin{enumerate}
\item Initialize $t_i = \max_{t \in T_i} \{t\}$, the highest un-ironed type.
\item For any $t \leq t_i$, define the average virtual value $a([t,t_i]) = \frac{\sum_{t' \in [t, t_i]}f_i(t')\cdot \varphi_i(t')}{\sum_{t' \in [t, t_i]} f_i(t')}$. 
\item Let $t^*_i$ maximize the average virtual value. That is, $t^*_i = \arg\max_{t \leq t_i} a([t, t_i])$ (break ties in favor of the maximum such $t^*_i$). 
\item Update $\tilde{\varphi}_i(t) = a([t^*_i, t_i])$ for all $t \in [t^*_i, t_i]$. 
\item Update $t \sim t'$ for all $t, t' \in [t^*_i, t_i]$. 
\item Update $t_i = \max_{t < t^*_i, t \in T_i}\{t\}$, the highest un-ironed type.
\item Return to Step 2.
\end{enumerate}
\end{definition}

These definitions in the discrete case might be slightly different than what readers are used to in the continuous case. We provide some observations proving that this is ``the right'' definition for the discrete case briefly in Section~\ref{sec:discrete}. Our proof continues in Section~\ref{sec:proof}.

\subsection{Discrete Myersonian Virtual Values}\label{sec:discrete}
In the continuous setting, Myerson's virtual valuation is defined as $\varphi_i(v) = v - \frac{1-F_i(v)}{f_i(v)}$, where $F_i$ and $f_i$ are the CDF and PDF of $D_i$. We first show that for any continuous distribution, discretizing it into multiples of $\epsilon$ and taking virtual valuations as in Definition~\ref{def:vv}, we recover Myerson's virtual valuation in the limt as $\epsilon \rightarrow 0$.

\begin{observation}\label{obs:discrete}
Let $D_i$ be any continuous distribution, and $D_i^\epsilon$ be the discretization of $D_i$ with point-masses at all multiples of $\epsilon$. That is, $f_i^\epsilon(c\epsilon) = \int_{c\epsilon}^{(c+1)\epsilon}f_i(x)dx$ for all $c \in \mathbb{N}$. Then for all $t_i$, $$\limsup_{\epsilon \rightarrow 0} \varphi^{D_i^\epsilon}_i(t_i) =t_i - \frac{1-F_i(t_i)}{f_i(t_i)}.$$
\end{observation}
\begin{proof}
For fixed $t_i$, consider the set of $\epsilon \in \{t_i/c\ |\ c \in \mathbb{N}\}$. Then clearly for all $\epsilon$ outside this set, $\varphi_i^{D_i^\epsilon}(t_i) = 0$. For any $\epsilon$ in this set, we have $\varphi_i^{D_i^\epsilon}(t_i) = t_i - \frac{\epsilon\cdot \Prob_{t\sim D_i}[t > t_i + \epsilon]}{\Prob_{t \sim D_i}\left[t \in [t_i,t_i+\epsilon]\right]}$. It's also clear that as $\epsilon \rightarrow 0$, we have $\Prob_{t \sim D_i}[t > t_i + \epsilon] \rightarrow \Prob_{t \sim D_i}[t > t_i] = 1-F_i(t_i)$, and $\frac{\epsilon}{\Prob_{t \sim D_i}[t \in [t_i,t_i+\epsilon]]}\rightarrow  \frac{1}{f_i(t_i)}$ (the latter is simply the definition of probability density). 
\end{proof}

A second valuable property of Myersonian virtual values is that they capture the ``marginal revenue.'' That is, if the seller was selling to a single bidder at price just above (i.e. $dv$ above) $v$, and decreased the price to just below (i.e. $dv$ below) $v$, the revenue would go up by exactly $\varphi(v)\cdot f(v) dv$.  We confirm that discrete virtual values as per Definition~\ref{def:vv} satisfy this property as well.
\begin{observation}\label{obs:discrete2}
For any single-dimensional discrete distribution $D_i$, we have $t_i\cdot \Prob_{t \sim D_i}[t \geq t_i] - t'_i\cdot  \Prob_{t \sim D_i}[t \geq t'_i] = f_i(t_i) \cdot \varphi_i(t_i)$, where $t'_i = \min_{t > t_i, t \in T_i}\{t\}$. In other words, $\varphi_i(t_i)$ captures the marginal change in revenue as we go from setting price $t'_i$ to price $t_i$. 
\end{observation}
\begin{proof}
This follows immediately from the definition of $\varphi_i(\cdot)$. But to be thorough:
\begin{align*}
t_i\cdot  &\Prob_{t \sim D_i}[t \geq t_i] &\\
&=t_i \cdot \left(\Prob_{t \sim D_i}[t \geq t'_i]+f_i(t_i)\right)&\text{(using that $t'_i = \min_{t > t_i, t \in T_i}\{t\}$)}\\
&=t'_i \cdot \left(\Prob_{t \sim D_i}[t \geq t'_i]+f_i(t_i)\right) - (t'_i - t_i)\cdot \left(\Prob_{t \sim D_i}[t \geq t'_i]+f_i(t_i)\right)&\\
&=t'_i \cdot \Prob_{t \sim D_i}[t \geq t'_i] + t_i \cdot f_i(t_i) - (t'_i - t_i)\cdot \Prob_{t \sim D_i}[t \geq t'_i]&\\
& = t'_i \cdot \Prob_{t \sim D_i}[t \geq t'_i] + f_i(t_i) \cdot \varphi_i(t_i)&\text{(definition of $\varphi_i(\cdot)$)}
\end{align*}
\end{proof}

We need one more definition specific to discrete type spaces before we can get back to the proof. A little more specifically: Myerson's payment identity, which shows that allocation rules uniquely determine payments for any BIC mechanism over continuous type spaces, does not apply for \emph{all} BIC mechanisms when types are discrete. Fortunately, the payment identity still holds for any mechanism that might possibly maximize revenue, but we need to be formal about this. 

\begin{definition} A BIC mechanism has \textbf{proper payments} if it is not possible to increase payments while keeping the allocation rule the same without violating BIC. Formally, a BIC mechanism $M = (\pi, p)$ has proper payments if for all $i$ and all subsets $S \subseteq T_i$, and all $\epsilon > 0$, increasing $p_i(t_i)$ by $\epsilon$ for all $t_i \in S$ while keeping $\pi$ the same does not result in a BIC mechanism. Note that all revenue-optimal mechanisms have proper payments.
\end{definition}
\begin{lemma}\label{lem:proper}
Let $\pi_i(\cdot)$ be monotone non-decreasing for all $i$. Then there exist $\{p_i(\cdot)\}_{i \in [n]}$ such that $M = (\pi, p)$ {is BIC} and has proper payments.
\end{lemma}
\begin{proof}
For ease of notation in the proof, label the types in $T_i$ so that $0 = t^0_i \leq t^1_i < \ldots < t^{|T_i|}_i$ ($t^j_i \in T_i$ for $j \in \{1,\ldots,|T_i|\}$). Then for a given $\pi_i(\cdot)$, define:
$$p_i(t^j_i) = \sum_{k=1}^j t^k_i\cdot (\pi_i(t^k_i) - \pi_i(t^{k-1}_i)).$$

We first claim that all $t^j_i$ are indifferent between telling the truth and reporting $t^{j-1}_i$. This is clear, as by definition of $p_i(\cdot)$ we have $p_i(t^j_i) = p_i(t^{j-1}_i) + t^j_i \cdot (\pi_i(t^j_i) - \pi_i(t^{j-1}_i))$, which implies:
$$t^j_i \cdot \pi_i(t^j_i) - p_i(t^j_i) = t^j_i \cdot \pi_i(t^j_i) - p_i(t^{j-1}_i) -t^j_i \cdot \pi_i(t^j_i) + t^j_i \cdot \pi_i(t^{j-1}_i) = t^j_i \cdot \pi_i(t^{j-1}_i) - p_i(t^{j-1}_i).$$

Now, consider any set $S \subseteq T_i$ and any $\epsilon > 0$, and imagine raising the payments of all types $t \in S$ by $\epsilon$. If we have $t^j_i \in S, t^{j-1}_i \not\in S$ for any $t^j_i$, then increasing all payments in $S$ by $\epsilon$ will cause $t^j_i$ to prefer reporting $t^{j-1}_i$ instead of telling the truth. So $S$ must contain all of $T_i$, and in particular $t^1_i$. But if we increase the payment of $t^1_i$ by $\epsilon > 0$, we violate individual rationality, as we defined $p_i(t^1_i) = t^1_i \cdot \pi_i(t^1_i)$.  So no such $S$, $\epsilon$ can exists, and $M = (\pi, p)$ has proper payments.

Finally, we just need to show that $M$ is BIC. Notice that by definition of $p_i(\cdot)$, for any $j > k$, we have $p_i(t^j_i) - p_i(t^k_i) \in \left[t^{k+1}_i \cdot (\pi_i(t^j_i)-\pi_i(t^k_i)), t^j_i\cdot (\pi_i(t^j_i)-\pi_i(t^k_i))\right]$. The lower bound corresponds to the case that all the change from $\pi_i(t^k_i)$ to $\pi_i(t^j_i)$ occurs going from $t^k_i$ to $t^{k+1}_i$ (i.e. $\pi_i(t^{k+1}_i) = \pi_i(t^j_i)$), and the upper bound corresponds to where all the change occurs going form $t^{j-1}_i$ to $t^j_i$ (i.e. $\pi_i(t^k_i) = \pi_i(t^{j-1}_i)$). 

The lower bound directly implies that $t^k_i$ prefers telling the truth to reporting $t^j_i$ as:
$$t^k_i \cdot (\pi_i(t^j_i) - \pi_i(t^k_i)) \leq p_i(t^j_i) - p_i(t^k_i) \Rightarrow t^k_i \cdot \pi_i(t^k_i) - p_i(t^k_i) \geq t^k_i \cdot \pi_i(t^j_i) - p_i(t^j_i).$$
Similarly, the upper bound directly implies that $t^j_i$ prefers telling the truth to reporting $t^k_i$ as:
$$t^j_i \cdot (\pi_i(t^j_i) - \pi_i(t^k_i)) \geq p_i(t^j_i) - p_i(t^k_i) \Rightarrow t^j_i \cdot \pi_i(t^j_i) - p_i(t^j_i) \geq t^j_i \cdot \pi_i(t^k_i) - p_i(t^k_i).$$
As the above holds for any $j > k$, $M$ is BIC.

\end{proof}

\subsection{Proof of Myerson's Theorem via Duality}\label{sec:proof}
Now we return to our proof of Myerson's Theorem. Let us first quickly confirm that indeed the resulting $\tilde{\varphi}_i(\cdot)$ by our ironing procedure is monotone:

\begin{observation}[\cite{Myerson81}]\label{obs:monotonevv} $t_i > t'_i \Rightarrow \tilde{\varphi}_i(t_i) \geq \tilde{\varphi}_i(t'_i)$. 
\end{observation}
\begin{proof}
First, all types in the same ironed interval share the same ironed virtual value. So if $\tilde{\varphi}_i(\cdot)$ is not monotone non-decreasing, there exists two adjacent ironed intervals $[x,y]$ and $[z,w]$ such that $x>w$ but $a([x,y])<a([z,w])$. Note that $a([z, y]) = c \cdot a([x, y])+(1-c) \cdot a([z, w])$ for some $c \in (0,1)$.\footnote{In fact, $c = \frac{\sum_{t \in [x, y]}f_i(t)}{\sum_{t \in [z, y]}f_i(t)}$.} Since $a([x,y])<a([z,w])$, we have $a([z, y])>a([x,y])$. However, this contradicts with the choice of an ironed interval for $y$ as specified in Step 3 of the ironing process. Hence, no such ironed intervals exist and $\tilde{\varphi}_i(\cdot)$ is monotone non-decreasing.
\end{proof}

And now, we can state Myerson's theorem applied to discrete type spaces. Afterwards, we will provide a proof using our new duality framework. One should map the theorem statement below to the statement of Theorem~\ref{thm:revenue less than virtual welfare} and see that our proof will essentially follow by providing a flow $\lambda$ that induces a virtual valuation function $\Phi^\lambda_i(\cdot) = \varphi_i(\cdot)$, another flow $\lambda'$ inducing $\Phi^{\lambda'}_i(\cdot) = \tilde{\varphi}_i(\cdot)$, and understanding which edges have non-zero flow in each. 

\begin{theorem}[\cite{Myerson81}]\label{thm:myerson}
For any BIC mechanism $M=(\pi, p)$, the revenue of $M$ is less than or equal to the virtual welfare of $\pi$ w.r.t. the virtual valuation function $\varphi(\cdot)$, and is less than or equal to the virtual welfare of $\pi$ w.r.t. the ironed virtual valuation function $\tilde{\varphi}(\cdot)$. That is:
\begin{align}
\sum_{i =1}^n \sum_{t_i \in T_i} f_i(t_i) \cdot p_i(t_i) \leq \sum_{i=1}^n \sum_{t_i \in T_i} f_i(t_i) \cdot \pi_i(t_i) \cdot \varphi_i(t_i).\label{eq:vw}\\
\sum_{i =1}^n \sum_{t_i \in T_i} f_i(t_i) \cdot p_i(t_i) \leq \sum_{i=1}^n \sum_{t_i \in T_i} f_i(t_i) \cdot \pi_i(t_i) \cdot \tilde{\varphi}_i(t_i).\label{eq:ivw}
\end{align}

Equality holds in Equation~\eqref{eq:vw} whenever $M$ has proper payments. Equality holds in Equation~\eqref{eq:ivw} if and only if $M$ has proper payments and $\pi_i(t_i) = \pi_i(t'_i)$ for all $t_i \sim t'_i$. Furthermore, the revenue-optimal BIC mechanism awards the item to the bidder with the highest non-negative ironed virtual value (if one exists), breaking ties arbitrarily but consistently across inputs.\footnote{Breaking ties consistently means that for any two different inputs, as long as they have the same ironed virtual value profile, the tie breaking should be the same.} If no such bidder exists, the item remains unallocated. 
\end{theorem}

We first provide a canonical flow inducing Myerson's virtual values as the virtual transformation. The first two lemmas below relate to proving Equation~\eqref{eq:vw}. The third relates to proving Equation~\eqref{eq:ivw}. 

\begin{lemma}\label{lem:myersonflow}
Define $\lambda_i(t'_i, t_i) = \Prob_{t \sim D_i}[t > t_i]$, where $t'_i = \min_{t > t_i, t \in T_i} \{t\}$, and $\lambda_i(t, t_i) = 0$ for all $t \neq t'_i$. Then $\lambda$ is a useful dual, and $\Phi_i^\lambda(t_i) = \varphi_i(t_i)$.
\end{lemma}
\begin{proof}
That $\lambda$ is a useful dual follows immediately by considering the total flow in and flow out of any given $t_i$. The total flow in is equal to $f_i(t_i) + \Prob_{t \sim D_i}[t > t_i]$. This is because $t_i$ receives flow $\Prob_{t \sim D_i} [t > t_i]$ from $t'_i$, and $f_i(t_i)$ from the super source. The total flow out is equal to $\Prob_{t \sim D_i} [t \geq t_i]$, so the two are equal.

To compute $\Phi_i^\lambda(t_i)$, simply plug the choice of $\lambda$ into Definition~\ref{def:virtual value}.
\end{proof}

\begin{lemma}\label{lem:tight}
In any BIC mechanism $M$ with proper payments, bidder $i$ with type $t_i$ is indifferent between reporting $t_i$ and $t'_i = \max_{t < t_i, t \in T_i}\{t\}$. 
\end{lemma}
\begin{proof}
We first recall that if $M = (\pi, p)$ is BIC, then $\pi_i(\cdot)$ and $p_i(\cdot)$ are both monotone non-decreasing for all $i$~\cite{Myerson81}. Because $M$ is BIC, we know that $t'_i \cdot \pi_i(t'_i) - p_i(t'_i) \geq t'_i \cdot \pi_i(t) - p_i(t)$ for all $t\in T_i$. Therefore, we also have:
$$t'_i \cdot (\pi_i(t'_i) - \pi_i(t)) \geq p_i(t'_i) - p_i(t).$$
By monotonicity, the LHS above is non-negative whenever $t'_i > t$. Therefore, $t_i \cdot(\pi_i(t'_i) - \pi_i(t)) \geq p_i(t'_i) - p_i(t)$ as well. This directly says that $t_i$ prefers reporting $t'_i$ to reporting any $t < t'_i$. Assume now for contradiction that there is some $t_i$ that is not indifferent between reporting $t_i$ and $t'_i$. Then we have the following chain of inequalities:
\begin{align*}
&t_i \cdot \pi_i(t_i) - p_i(t_i) > t_i \cdot \pi_i(t'_i) - p_i(t'_i) \geq t_i \cdot \pi_i(t') - p_i(t'),~\forall t' < t_i\\
\Rightarrow &t_i \cdot (\pi_i(t_i) - \pi_i(t')) > p_i(t_i) - p_i(t'),~\forall t'\leq t'_i\\
\Rightarrow &t \cdot (\pi_i(t_i) - \pi_i(t')) > p_i(t_i) - p_i(t'),~\forall t' < t_i, t \geq t_i\\
\Rightarrow &t \cdot \pi_i(t) - p_i(t) \geq t \cdot \pi_i(t_i) - p_i(t_i) > t \cdot \pi_i(t') - p_i(t'),~\forall t' < t_i, t \geq t_i.
\end{align*}
The last line explicitly states that all $t\geq t_i$ strictly prefer telling the truth to reporting any $t' \leq t'_i$. Therefore, there exists a sufficiently small $\epsilon > 0$ such that if we increase all $p_i(t)$ by $\epsilon$ for all $t \geq t_i$ then $M$ remains BIC, contradicting that $M$ has proper payments. 
\end{proof}

At this point, we have proved Equation~\eqref{eq:vw} and the related statements for unironed virtual values (but we will wrap up concretely at the end of the section). We now turn to Equation~\eqref{eq:ivw}, and first show how to ``fix'' non-monotonicities in virtual values by adding cycles.

\begin{lemma}\label{lem:cycles}
Starting from any flow $\lambda$ and induced virtual values $\Phi^\lambda$, adding a cycle of flow $x$ between bidder $i$'s type $t_i$ and type $t'_i < t_i$:
\begin{itemize}
\item Increases $\Phi_i^\lambda(t_i)$.
\item Decreases $\Phi_i^\lambda(t'_i)$.
\item Preserves 
$f_i(t_i)\cdot \Phi_i^\lambda(t_i)+f_i(t'_i)\cdot \Phi_i^\lambda(t'_i)$. 
\end{itemize}
\end{lemma}
\begin{proof}
Recall the definition of $\Phi^\lambda_i(t_i) = t_i +\frac{\sum_{t}(t_i - t) \lambda_i(t, t_i)}{f_i(t_i)}$. Adding a cycle of flow $x$ between $t_i$ and $t'_i$ increases $\lambda_i(t'_i, t_i)$ and $\lambda_i(t_i, t'_i)$, but otherwise doesn't change $\lambda$. So $\Phi^\lambda_i(t_i)$ increases by $(t_i - t'_i)x/f_i(t_i)$. Similarly, $\Phi_i^\lambda(t'_i)$ decreases by exactly $(t_i - t'_i)x/f_i(t'_i)$. It is therefore clear that $f_i(t_i)\cdot \Phi_i^\lambda(t_i)+f_i(t'_i)\cdot \Phi_i^\lambda(t'_i)$ preserved, as the changes are inversely proportional to the types' densities. 
\end{proof}
\noindent And we may now conclude that a flow exists inducing Myerson's ironed virtual values as well.

\begin{corollary}\label{cor:cycles}
There exists a flow $\lambda$ such that:
\begin{itemize}
\item $\Phi^\lambda_i(t_i) = \tilde{\varphi}_i(t_i)$. 
\item $\lambda_i(t_i, t'_i) > 0$ whenever $t'_i = \max_{t < t_i, t \in T}\{t\}$. 
\item $\lambda_i(t'_i, t_i) > 0$ when $t'_i = \max_{t < t_i, t \in T}\{t\}$ \textbf{if and only if} $t_i$ and $t'_i$ are in the same ironed interval. 
\item $\lambda_i(t, t') = 0$ for all other $t, t'$.
\end{itemize}
\end{corollary}
\begin{proof}
Start with the $\lambda$ inducing $\Phi^\lambda_i(t_i) = \varphi_i(t_i)$, and consider any ironed interval $[x, y]$. First, observe that we must have $\varphi_i(x) > a(x,y)$, as otherwise $[x, y]$ would not be an ironed interval (as $x$ doesn't maximize $a(x,y)$, $a(z,y)$ would be at least as large where $z = \min_{t > x, t \in T_i}\{t\}$, and recall that we would break a tie in favor of $z$). 

Now, add a cycle between $x$ and $z = \min_{t > x, t \in T_i}\{t\}$. Per Lemma~\ref{lem:cycles}, this increases $\Phi^\lambda_i(z)$ and decreases $\Phi^\lambda_i(x)$. So increase the weight along this cycle until $\Phi^\lambda_i(x)$ decreases to $\tilde{\varphi}_i(x)$. At this point, either $[x,z]$ is the entire ironed interval, in which case this interval is ``finished.'' Or, maybe $z < y$. In this case, we claim that we can iterate the process with $z$. To see this, observe that by Lemma~\ref{lem:cycles}, we must have preserved $a(x, z)$. Observe again that we must have $a(x,z) > a(x,y)$ in order for $[x,y]$ to be an ironed interval, and we have just set $\Phi^\lambda_i(x) = \tilde{\varphi}_i(x) = a(x,y) < a(x,z)$, so we must have $\Phi^\lambda_i(z) > a(x,z) > a(x,y) = \tilde{\varphi}_i(z)$. So we can again add a cycle between $z$ and $w = \min_{t > z, t \in T_i} \{t\}$ to decrease $\Phi^\lambda_i(z)$ to $\tilde{\varphi}_i(z)$ while preserving $a(x, w)$. Iterating this process all the way until $y$ necessarily adds a cycle between all adjacent types and results in $\Phi^\lambda_i(t) = \tilde{\varphi}_i(t)$ for all $t \in [x,y]$, again by Lemma~\ref{lem:cycles}. Repeating this argument for all ironed intervals proves the corollary.
\end{proof}

\noindent Now we may complete the proof of Theorem~\ref{thm:myerson}. The bulk of the proof is captured by Lemma~\ref{lem:myersonflow} and Corollary~\ref{cor:cycles}, the only remaining work is to confirm the structure of the optimal mechanism.\\

\begin{prevproof}{Theorem}{thm:myerson}
Inequality~\eqref{eq:vw} now immediately follows from Theorem~\ref{thm:revenue less than virtual welfare} and Lemma~\ref{lem:myersonflow}, and the condition for it to be an equality is implied by Lemma~\ref{lem:tight}.  Inequality~\eqref{eq:ivw} follows immediately from Theorem~\ref{thm:revenue less than virtual welfare} and Corollary~\ref{cor:cycles}. Next, we argue why Inequality~\eqref{eq:ivw} is an equality when the stated condition holds. When a mechanism is BIC and $\pi_i(t_i)=\pi_i(t_i')$, then $p_i(t_i)=p_i(t'_i)$, because $0=t_i\cdot(\pi_i(t_i)-\pi_i(t'_i))\geq p_i(t_i)-p_i(t'_i)\geq t'_i\cdot(\pi_i(t_i)-\pi_i(t'_i))=0$ (implied by the BIC constraints). Therefore, when the condition holds, any bidder $i$ with type $t_i$ is indifferent between reporting $t_i$ and any type $t'_i\sim t_i$. Combining this observation with Lemma~\ref{lem:tight}, we know that for the flow specified in Corollary~\ref{cor:cycles}, all BIC constraints bind between any two types $t_i$ and  $t'_i$ with $\lambda(t_i,t'_i)>0$. Thus, Inequality~\eqref{eq:ivw} is an equality when the condition holds due to Theorem~\ref{thm:revenue less than virtual welfare}.

Finally, to see that the optimal mechanism has the prescribed format, observe that the allocation rule that awards the item to the highest non-negative ironed virtual value clearly maximizes ironed virtual welfare. So we get that the revenue of the optimal BIC mechanism is upper bounded by the ironed virtual welfare of this allocation rule. Moreover, ironed virtual values are always monotone non-decreasing, so by Lemma~\ref{lem:proper}, this allocation rule has corresponding proper payments that combine to a BIC mechanism. Finally, because the allocation rule by definition satisfies $\pi_i(t_i) = \pi_i(t'_i)$ whenever $t_i \sim t'_i$, we have that the revenue of this mechanism is equal to its expected ironed virtual welfare (again by Theorem~\ref{thm:revenue less than virtual welfare}), and is therefore optimal (as its expected ironed virtual welfare is an upper bound on the expected revenue of any BIC mechanism). 
\end{prevproof}

In summary, we have provided a duality-based proof of Myerson's Theorem~\cite{Myerson81}. This gives some intuition for the flows we will develop in the following section. Also, it provides a different insight into the difference between ironed and non-ironed virtual values. Expected revenue is equal to expected virtual welfare for all BIC mechanisms (with proper payments) because the flow necessary to derive virtual values only sends non-zero flow along edges that correspond to BIC constraints that are always tight. On the other hand, revenue is only upper bounded by ironed virtual welfare for all BIC mechanisms because the flow necessary to derive ironed virtual values sends non-zero flow along all edges between adjacent types in an equivalence class. So revenue is only equal to ironed virtual welfare if all of the corresponding BIC constraints are tight (and there exist truthful mechanisms for which this doesn't hold). 

%% file: flow.tex

\section{Canonical Flow and Virtual Valuation Function for Multiple Items}\label{sec:flow}

In this section, we present a canonical way to set the Lagrangian multipliers/flow that induces our benchmarks for multi-item settings. This flow will use similar ideas to Section~\ref{sec:myerson}. Informally, our approach for a single bidder first divides the entire type space of the bidder into regions based on their \emph{favorite item} (that is, $\argmax_j \{v_j\}$). We'll then use a different ``Myerson-like'' flow within each region, described in more detail shortly. For multiple bidders, we'll still divide the type space of each bidder into regions based on their favorite item, but define the ``favorite'' item slightly differently.

Specifically, let $P_{ij}(t_{-i})$ denote the price that bidder $i$ could pay to receive exactly item $j$ in the VCG mechanism against bidders with types $t_{-i}$.\footnote{Note that when buyers are additive, this is exactly the highest bid for item $j$ from buyers besides $i$. When buyers are unit-demand, buyer $i$ only ever buys one item, and this is the price she would pay for receiving $j$.}
We will partition the type space $T_{i}$ into $m+1$ regions:
\textbf{(i)} $R^{(t_{-i})}_{0}$ contains all types $t_{i}$ such that $t_{ij}< P_{ij}(t_{-i})$, $\forall j$; \textbf{(ii)} $R^{(t_{-i})}_{j}$ contains all types $t_{i}$ such that $t_{ij}-P_{ij}(t_{-i})\geq 0$ and $j$ is the smallest index in $\argmax_{k}\{t_{ik}-P_{ik}(t_{-i})\}$. 
This partitions the types into subsets based on which item provides the largest surplus (value minus price), and we break ties lexicographically. We'll refer to the largest surplus item as the bidder's favorite item. We'll refer to all other items as \emph{non-favorite} items. For any bidder $i$ and any type profile $t_{-i}$ of everyone else, we define $\lambda^{(t_{-i})}_{i}$ to be the following flow. The flow as defined below will look similar to Myerson's (non-ironed) virtual values, and we will need to similarly iron it to accommodate irregular distributions.

For the remainder of this section, it will be helpful to have the above concrete definition of $R_j^{(t_{-i})}$ in mind for defining our flows. However, all results proved in this section apply more broadly, for any definitions of $R_j^{(t_{-i})}$ which are \emph{upwards-closed}. Subsequent work (e.g.~\cite{EdenFFTW17a}) has made use of these generalized benchmarks.

{\begin{definition}[Upwards Closed Regions] We say that regions $\{R_j^{(t_{-i})}\}_{j \in [m]}$ are \emph{upwards-closed} if for all $j\geq 1$, whenever ${v} \in R_j^{(t_{-i})}$, ${v} + c \cdot e_j \in R_j^{(t_{-i})}$ for any $c \geq 0$ as well. Here, $e_j$ denotes the $j^{th}$ standard basis vector in the $m$-dimensional Euclidean space. 
\end{definition}}

{In this section, we state/prove our results for arbitrary upwards-closed regions. In subsequent sections, we will instantiate $R_j^{(t_{-i})}$ as defined in the previous paragraphs.}
\begin{definition}[Initial Canonical Flow] Define our initial canonical flow to be the following:
 \begin{enumerate}
 \item  $\forall j>0$, any flow entering $R^{(t_{-i})}_{j}$ is from  $s$ (the super source) and any flow leaving $R^{(t_{-i})}_{j}$ is to $\varnothing$ (super sink).
 \item For every type $t_{i}$ in region $R^{(t_{-i})}_{0}$,  the flow goes directly to $\varnothing$. That is, $\lambda_i(t_i, \varnothing) = f_i(t_i)$ for all $t_i \in R^{(t_{-i})}_0$. 
\item For every type $t_i$ in region $R^{(t_{-i})}_j$, define type $t'_i$ such that $t'_{ij} = \max_{t < t_{ij}, t \in T_{ij}}\{t\}$, and $t'_{ik} = t_{ik}$ for all $k \neq j$. 
\begin{itemize}
\item If $t'_i \in R^{(t_{-i})}_j$ as well, then set $\lambda_i(t_i, t'_i) = \sum_{t^*_i:\  t^*_{ik} = t_{ik}\ \forall k \neq j\  \land\ t^*_{ij} \geq t_{ij}}f_i(t^*_i)$.
\item If $t'_i \notin R^{(t_{-i})}$, set $\lambda_i(t_i, \varnothing) =  \sum_{t^*_i:\  t^*_{ik} = t_{ik}\ \forall k \neq j\  \land\ t^*_{ij} \geq t_{ij}}f_i(t^*_i)$. 
\end{itemize}
\end{enumerate}
\end{definition}

Observe that this indeed defines a flow. All nodes in $R^{(t_{-i})}_0$ have no flow in (except from the super source), and send exactly this flow to the super sink. All other nodes get flow in from exactly one type (in addition to the super source) and send flow out to exactly one type, and the flow is balanced, just as in Lemma~\ref{lem:myersonflow}. At a high level, what we are doing is restricting attention to a single item and attempting to use the canonical single item flow for just this item. We restrict attention to different items for different types, depending on {the region in which $t_i$ lies (for our later instantiation, this depends on which item gives them highest utility at the prices $P_{ij}(t_{-i})$)}. Let's first study the induced virtual values $\Phi_{ik}^{(t_{-i})}(t_i)$ for when $t_i \notin R_k^{(t_{-i})}$. 

\begin{claim}\label{clm:value term}
For any type $t_{i}\in R^{(t_{-i})}_{j}$, its corresponding virtual value $\Phi_{ik}^{(t_{-i})}(t_{i})$ for item $k$ is exactly its value $t_{ik}$ for all $k\neq j$.
\end{claim}

\begin{proof}
By the definition of $\Phi_{i}^{(t_{-i})}(\cdot)$ , $\Phi_{ik}^{(t_{-i})}(t_{i})= t_{ik}-{1\over f_{i}(t_{i})}\sum_{t_{i}'} \lambda^{(t_{-i})}_{i}(t_{i}',t_{i})(t_{ik}'-t_{ik})$. Since $t_{i}\in R_{j}$, by the definition of the flow $\lambda^{(t_{-i})}_{i}$, for any $t_{i}'$ such that $\lambda^{(t_{-i})}_{i}(t_{i}',t_{i})>0$, $t_{ik}'-t_{ik}=0$ for all $k\neq j$, therefore $\Phi_{ik}^{(t_{-i})}(t_{i})= t_{ik}$.
\end{proof}

Let's now study the corresponding $\Phi_{ij}^{(t_{-i})}(t_i)$ for this flow when $t_i \in R^{(t_{-i})}_j$. This turns out to be closely related to the Myerson's virtual value function for single-dimensional distributions discussed in Section~\ref{sec:myerson}. For each $i,j$, we use $\varphi_{ij}(\cdot)$ and $\tp_{ij}(\cdot)$ to denote the Myerson virtual value and ironed virtual value function for distribution $D_{ij}$ respectively, as defined in Section~\ref{sec:myerson}. 

\begin{claim}\label{clm:regular flow}
For any type $t_{i}\in R^{(t_{-i})}_{j}$, then the initial canonical flow induces virtual values satisfying: $\Phi_{ij}^{(t_{-i})}(t_{i})=\varphi_{ij}(t_{ij})=t_{ij}-\frac{{(t'_{ij}-t_{ij})\cdot}\Pr_{t\sim D_{ij}}[t>t_{ij}]}{f_{ij}(t_{ij})}$, where $t'_{ij} = \min_{t > t_{ij}, t \in T_{ij}}\{t\}$. 
\end{claim}
\begin{proof}
Let us fix $t_{i,-j}$, and prove this is true for all choices of $t_{i,-j}$. If $t_{ij}$ is the largest value in $T_{ij}$, then there is no flow coming into it except the one from the source, so $\Phi_{ij}^{(t_{-i})}(t_{i})=t_{ij}$. For every other value of $t_{ij}$, the flow coming from its predecessor $(t'_{ij},t_{i,-j})$ is exactly (note below that several steps make use of the fact that $f_i(t_i) = \prod_j f_{ij}(t_{ij})$)
$$\sum_{t^*_i:\ t^*_{ik}=t_{ik} \forall k \neq j\ \land\  t^*_{ij} \geq t'_{ij}}f_i(t_i^*) = \sum_{t^*_i:\ t^*_{ik}=t_{ik} \forall k \neq j\ \land\  t^*_{ij} \geq t'_{ij}}f_{ij}(t^*_{ij})\cdot\prod_{k \neq j} f_{ik}(t_{ik}) = \prod_{k\neq j} f_{ik}(t_{ik})\cdot \sum_{t_{ij}^*>t_{ij}} f_{ij}(t_{ij}^*) $$
$$= \prod_{k\neq j} f_{ik}(t_{ik})\cdot \Pr_{t \sim D_{ij}}[t>t_{ij}].$$
Now, we can compute according to Definition~\ref{def:virtual value}: 

$$\Phi_{ij}^{(t_{-i})}(t_{i})= t_{ij} - \frac{(t'_{ij}-t_{ij})\cdot \prod_{k \neq j}f_{ik}(t_{ik})\cdot \Pr_{t \sim D_{ij}}[t > t_{ij}]}{f_i(t_i)} = t_{ij} - \frac{(t'_{ij}-t_{ij})\cdot \Pr_{t \sim D_{ij}}[t > t_{ij}]}{f_{ij}(t_{ij})} =\varphi_{ij}(t_{ij})$$
\end{proof}

Claims~\ref{clm:value term} and~\ref{clm:regular flow} show that our initial canonical flow induces virtual values such that the virtual value of each bidder for all of their \emph{non-favorite} items is exactly their value, while their virtual value for their \emph{favorite} item is exactly their Myersonian virtual value as per Definition~\ref{def:vv}. When $D_{ij}$ is regular, this is the canonical flow we use. When the distribution is not regular, we also need to ``iron'' the virtual values as in Section~\ref{sec:myerson}. Essentially all we are doing is applying the same procedure as Definition~\ref{def:ironing}, but we repeat it below to be clear exactly how the substitutions occur. Below, we use $\tp^{*}_{ij}$ to denote the ironed virtual values instead of $\tp_{ij}$ because we reserve $\tp_{ij}$ to refer exactly to the ironed virtual values that result in the single item case, and we haven't yet proved that they are (essentially) the same.

\begin{definition}[Ironed Canonical Virtual Values]\label{def:multiiron} For a given bidder $i$, valuation vector $t_i \in R_j^{(t_{-i})}$, obtain the ironed canonical values, $\tilde{\varphi}^{*}_{ij}(\cdot)$ in the following manner: let $X$ denote the minimum $t_{ij}$ such that $(t_{ij}; t_{i, -j}) \in R^{(t_{-i})}_j$. Let $\sim^*$ be an equivalence relation on $T_{ij}$. We say that an interval $[t^*_{ij}, t_{ij}]$ is ironed if $t \sim^* t'$ for all $t, t' \in [t^*_{ij}, t_{ij}]$. 
\begin{enumerate}
\item Initialize $t_{ij} = \max_{t \in T_{ij}} \{t\}$, the highest un-ironed type.
\item For any $t\in [X,t_{ij}]$, define the average virtual value $a([t,t_{ij}]) = \frac{\sum_{t' \in [t, t_{ij}]}f_{ij}(t')\cdot \varphi_{ij}(t')}{\sum_{t' \in [t, t_{ij}]} f_{ij}(t')}$. 
\item Let $t^*_i$ maximize the average virtual value. That is, $t^*_{ij} = \arg\max_{t \in [X, t_{ij}] } \{a([t, t_{ij}])\}$ (break ties in favor of the maximum such $t^*_{ij}$). 
\item Update $\tilde{\varphi}^{*}_{ij}(t) = a([t^*_{ij}, t_{ij}])$ for all $t \in [t^*_{ij}, t_{ij}]$. 
\item Update $t \sim^* t'$ for all $t, t' \in [t^*_{ij}, t_{ij}]$. 
\item Update $t_{ij}= \max_{t \in [X, t^*_{ij}), t \in T_i}\{t\}$, the highest un-ironed type.
\item Return to Step 2.
\end{enumerate}

\end{definition}

\begin{lemma}\label{lem:multicycles}
The initial canonical flow can be ironed into a $\lambda$ so that for all $j$, and all $t_i \in R^{(t_{-i})}_j$, $\Phi^\lambda_{ij}(t_i) = \tilde{\varphi}^{*}_{ij}(t_{ij})$ by only adding cycles between types $t_i, t'_i \in R^{(t_{-i})}_j$ satisfying $t_{i,-j} = t'_{i,-j}$.
\end{lemma}
\begin{proof}
Exactly the same as Corollary~\ref{cor:cycles}, plus the observation that all types for which values for item $j$ have identical values for items $\neq j$. 
\end{proof}

Now that we have a flow ``ironing'' one of the virtual values, we want to wrap up by observing that $\tilde{\varphi}^{*}_{ij}(v) \leq \tilde{\varphi}_{ij}(v)$ for all $v$ (where $\tilde{\varphi}_{ij}(v)$ is Myerson's ironed virtual value for the distribution $D_{ij}$. 

\begin{lemma}\label{lem:ironed virtual value} For any $i,j$, $t_{-i}, t_{i,-j}, t_{ij}$, $\tp^{*}_{ij}(t_{ij}) \leq \tp_{ij}(t_{ij})$.
\end{lemma}

\begin{proof}
Observe that the ironing procedure in Definition~\ref{def:multiiron} is nearly identical to that of Definition~\ref{def:ironing}. In fact, for all $x,y \geq X$ ($X$ is defined in Definition~\ref{def:multiiron}), if $x \sim^* y$ as in Definition~\ref{def:multiiron}, then $x \sim y$ as in Definition~\ref{def:ironing}. This immediately yields that $\tp_{ij}(t_{ij}) = \tp^{*}_{ij}(t_{ij})$ for all $t_{ij}$ in ironed intervals $[x,y]$ (as in Definition~\ref{def:ironing}) such that $x \geq X$. But the ironed virtual values might differ if $t_{ij}$ lies inside an ironed interval that is ``cut'' by $X$ in Definition~\ref{def:ironing}. But observe that by the definition of ironing, we necessarily have $a([x,y]) \geq a([X,y])$ in order for $[x,y]$ to possibly be an ironed interval containing $X$. Therefore, we may immediately conclude that $\tp_{ij}(t_{ij}) \geq \tp_{ij}^{*}(t_{ij})$ for all $t_{ij}$, even those in ironed intervals cut by $X$. 
\end{proof}

\begin{lemma}\label{lem:flow properties}
{Let $\{R_j^{(t_{-i})}\}_{j \in [m]}$ define upwards-closed regions. }Then there exists a flow $\lambda^{(t_{-i})}_{i}$ such that $\Phi_{ij}^{(t_{-i})}(t_{i})$ satisfies the following properties: 
\begin{itemize}
\item For any $j>0$, $t_{i}\in R^{(t_{-i})}_{j}$,  $\Phi_{ij}^{(t_{-i})}(t_{i})\leq \tp_{ij}(t_{ij})$, where $\tp_{ij}(\cdot)$ is Myerson's ironed virtual value for  $D_{ij}$. 
\item For any $j$, $t_{i}\in R^{(t_{-i})}_{j}$, $\Phi_{ik}^{(t_{-i})}(t_{i})=t_{ik}$ for all $k\neq j$. In particular, $\Phi_{i}^{(t_{-i})}(t_{i}) = t_{i}$, $\forall t_{i}\in R_{0}^{(t_{-i})}$.
\end{itemize}
\end{lemma}

\begin{proof}
To see the first bullet, combine Lemmas~\ref{lem:multicycles} and~\ref{lem:ironed virtual value}. To see the second bullet, observe Claim~\ref{clm:value term}, combined with the fact that the cycles added via Lemma~\ref{lem:multicycles} doesn't affect virtual values for the non-favorite items.
\end{proof}

\begin{corollary}\label{cor:upwardsclosed} Let $\{R_j^{(t_{-i})}\}_{j \in [m]}$ define upwards-closed regions. Then for any BIC mechanism with $(\pi,p)$ as its reduced form: 

$$\sum_i \sum_{t_i \in T_i} f_i(t_i) \cdot p_i(t_i) \leq \sum_i \sum_{t_i \in T_i} \sum_j f_i(t_i) \cdot \pi_{ij}(t_i) \cdot \left(t_{ij} \cdot \I\left[t_i \notin R_j^{(t_{-i})}\right] + \tilde{\varphi}_{ij}(t_{ij}) \cdot \I\left[t_i \in R_j^{(t_{-i})}\right]\right)$$
\end{corollary}

{Corollary~\ref{cor:upwardsclosed} upper bounds the optimal revenue for arbitrary upwards-closed regions. Corollary~\ref{cor:upwardsclosedadditive} further relaxes this upper bound by considering the $\pi$ which maximizes expected Virtual Welfare for additive bidders. Note that the relaxation below has been used in follow-up works (e.g.~\cite{EdenFFTW17a}) for additive bidders, but is a very loose relaxation for unit-demand bidders.}

\begin{corollary}\label{cor:upwardsclosedadditive}{Let $\{R_j^{(t_{-i})}\}_{j \in [m]}$ define upwards-closed regions. Then: 
$$\rev(D) \leq \mathbb{E}_{\vec{v} \sim D}\left[\sum_j \max_{i} \left\{t_{ij} \cdot \I\left[t_i \notin R_j^{(t_{-i})}\right] + \tilde{\varphi}_{ij}(t_{ij}) \cdot \I\left[t_i \in R_j^{(t_{-i})}\right]\right\}\right].$$} 
\end{corollary}

{Now we instantiate the specific choice of regions and flows for our canonical flow.} At this point, for each $t_{-i}\in T_{-i}$, we have defined a different flow for bidder $i$. We have shown that this flow induces a virtual valuation function such that bidder $i$'s virtual value for all non-favorite items is equal to their value for those items, and their virtual value for their favorite item is at most their Myersonian ironed virtual value for that item. Figure~\ref{fig:singleflow} contains a diagram illustrating our flow for a single bidder, and Figure~\ref{fig:multiflow} illustrates what the flow might look like for non-zero $(t_{-i})$. 

\begin{figure}
  \centering{\includegraphics[width=0.45\linewidth]{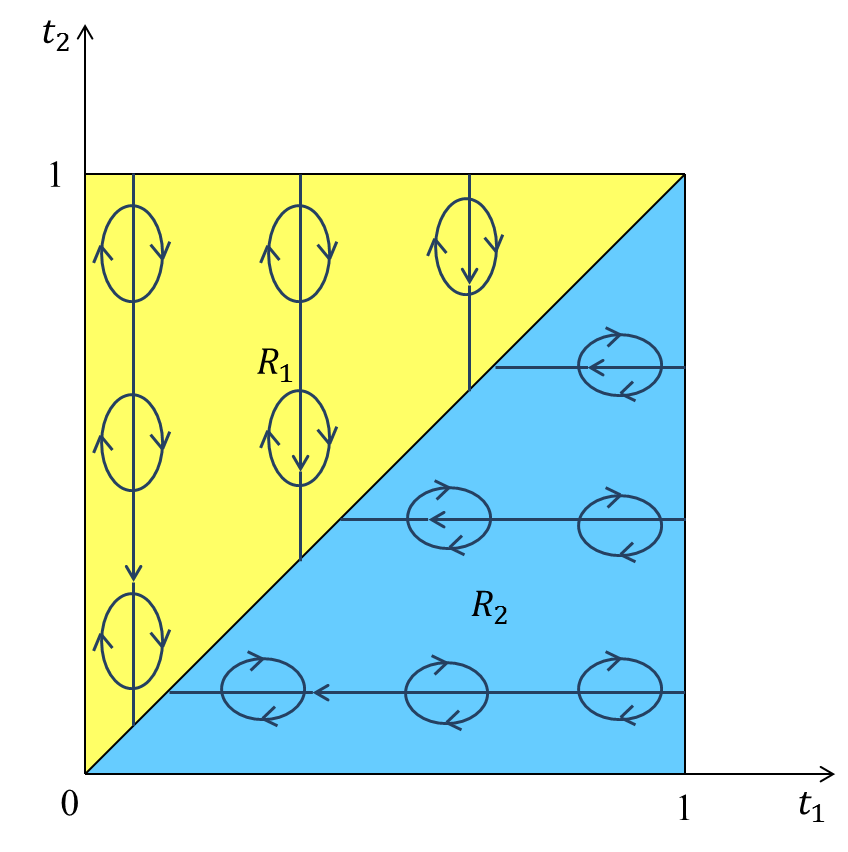}}
  \caption{An example of $\lambda$ (with ironing) for a single bidder.}
  \label{fig:singleflow}
\end{figure}

\begin{figure}
  \centering{\includegraphics[width=0.45\linewidth]{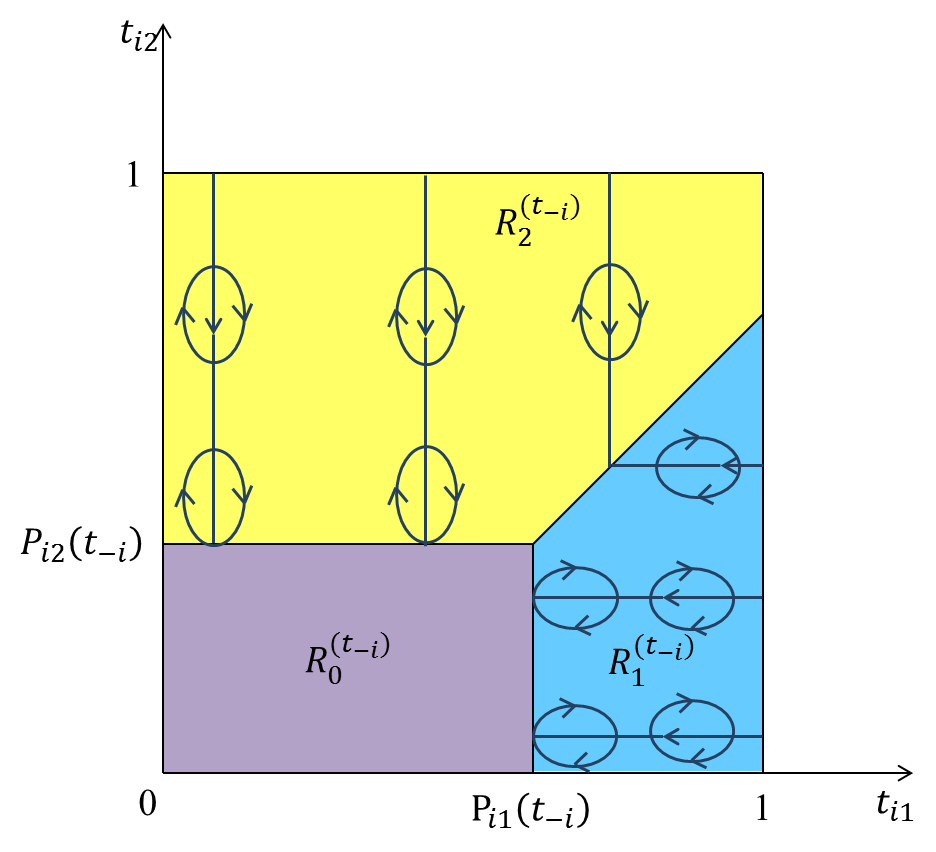}}
  \caption{An example of $\lambda^{(t_{-i})}_{i}$ for two items.}
  \label{fig:multiflow}
\end{figure}

Finally, note that we've defined many possible flows for bidder $i$: each $t_{-i}$ defines different $P_{ij}(t_{-i})$s, which in turn define different $R_j^{(t_{-i})}$s, which define different flows. But we only get to pick one flow for bidder $i$, and it cannot change depending on $t_{-i}$. The flow that we will finally use essentially averages these flows according to $D_{-i}$. 
   
\begin{definition}[Canonical Flow for Multiple Items] Our flow for bidder $i$ is $\lambda_{i}=\sum_{t_{-i}\in T_{-i}} f_{-i}(t_{-i})\lambda^{({t_{-i})}}_{i}$. Accordingly, the virtual value function $\Phi_{i}$ of $\lambda_{i}$ is $\Phi_{i}(\cdot)=\sum_{t_{-i}\in T_{-i}} f_{-i}(t_{-i}) \Phi_{i}^{(t_{-i})}(\cdot)$.
\end{definition}

\noindent\textbf{Intuition behind Our Flow:} The social welfare is a trivial upper bound for revenue, which can be arbitrarily bad in the worst case. To design a good benchmark, we want to replace some of the terms that contribute the most to the social welfare with more manageable ones. The flow $\lambda_i^{(t_{-i})}$ aims to achieve exactly this. For each bidder $i$, we find the item $j$ that contributes the most to the social welfare when awarded to $i$. Then we turn the \textbf{virtual value} of item $j$ into its Myerson's single-dimensional (ironed) virtual value, and keep the \textbf{virtual value} of all the other items equal to the value. This transformation is feasible only if we know exactly $t_{-i}$ and could use a different dual solution for each $t_{-i}$. Since we can't, a natural idea is to define a flow by taking an expectation over $t_{-i}$. This is indeed our flow. 

{We conclude this section with one final lemma and our main theorem regarding the canonical flow. Both proofs are immediate corollaries of the flow definition and Theorem~\ref{thm:revenue less than virtual welfare}.} Note also that our flow only ever sends flow between types that are identical on all but one coordinate, and adjacent in the final coordinate (and that this coordinate is their ``favorite'' item - adjusted by $t_{-i}$). This means that our benchmark not only upper bounds the optimal revenue of any BIC mechanism, but it also upper bounds the optimal revenue of any (non-truthful) mechanism where bidder $i$ with type $t_i$ has no incentive to lie by misreporting their value for a single item to an adjacent value. A corollary of our work in the following sections is that this relaxation does not improve the optimal revenue by more than a constant factor. 

\begin{lemma}\label{lem:value plus virtual value}
For all $i$, $j$, $t_{i}$, $\Phi_{ij}(t_{i})\leq t_{ij}\cdot\Pr_{v_{-i}\sim D_{-i}}\left[t_{i}\notin R^{(v_{-i})}_{j}\right]+\tp_{ij}(t_{ij})\cdot \Pr_{v_{-i}\sim D_{-i}}\left[t_{i}\in R^{(v_{-i})}_{j}\right]$.
\end{lemma}

\begin{theorem}\label{thm:virtual welfare ub}
Let $M$ be any BIC mechanism with $\big(\pi, p\big)$ as its reduced form. The expected revenue of $M$ is upper bounded by the expected virtual welfare of the same allocation rule with respect to the canonical virtual value function $\Phi_{i}(\cdot)$. In particular, 
\begin{align*}
\sum_{i}&\sum_{t_{i}\in T_{i}} f_{i}(t_{i})\cdot p_{i}(t_{i})
\leq\sum_{i}\sum_{t_{i}\in T_{i}} \sum_{j} f_{i}(t_{i})\cdot \pi_{ij}(t_{i})\cdot \Phi_{ij}(t_{i})\\
&\leq \sum_{i}\sum_{t_{i}\in T_{i}} \sum_{j} f_{i}(t_{i})\cdot \pi_{ij}(t_{i})\cdot 
\left(t_{ij}\cdot\Pr_{v_{-i}\sim D_{-i}}\left[t_{i}\notin R^{(v_{-i})}_{j}\right]
+\tp_{ij}(t_{ij})\cdot \Pr_{v_{-i}\sim D_{-i}}\left[t_{i}\in R^{(v_{-i})}_{j}\right]\right)  \stepcounter{equation}\tag{\theequation} \label{eq:ub} 
\end{align*}
\end{theorem}

%% file: single_bidder.tex

\section{Warm Up: Single Bidder}\label{sec:single}
As a warm up, we start with the single bidder case. In this section, our goal is to show how to use the bounds obtained via Theorem~\ref{thm:virtual welfare ub} to prove that simple mechanisms are approximately optimal for a single additive or unit-demand bidder with independent item values, recovering results of~\cite{ChawlaMS10} and~\cite{BabaioffILW14}. Throughout this section, we keep the same notations but drop the subscript $i$ and superscript $(t_{-i})$ whenever is appropriate.
 
\paragraph{Canonical Flow for a Single Bidder.}
Since the canonical flow and the corresponding virtual valuation functions are defined based on other bidders types $t_{-i}$, let us see how it is simplified when there is only a single bidder. First, the VCG prices are all $0$, therefore $\lambda$ is simply one flow instead of a distribution of different flows. Second, for the same reason, the region $R_{0}$ is empty and region $R_{j}$ contains all types $t$ with $t_{j}\geq t_{k}$ for all $k$ (see Figure~\ref{fig:singleflow} for an example). 
This simplifies Expression~(\ref{eq:ub}) to 
\begin{align*}
&\sum_{t\in T} \sum_{j} f(t)\cdot \pi_{j}(t)\cdot \Big(t_{j}\cdot \I[t\notin R_{j}]+\tp_{j}(t_{j})\cdot \I[t\in R_{j}]\Big)\\
=& \sum_{t\in T} \sum_{j} f(t)\cdot \pi_{j}(t)\cdot t_{j}\cdot \I[t\notin R_{j}] \quad (\textsc{Non-Favorite})\\
&~~~~~+ \sum_{t\in T} \sum_{j} f(t)\cdot \pi_{j}(t)\cdot \tp_{j}(t_{j})\cdot \I[t\in R_{j}]\quad (\textsc{Single})
\end{align*}
\notshow{\begin{align*}
&\sum_{t\in T} \sum_{j} f(t)\cdot \pi_{j}(t)\cdot \Big(t_{j}\cdot \I[t\notin R_{j}]+\tp_{j}(t_{j})\cdot \I[t\in R_{j}]\Big)\\
=& \sum_{t\in T} \sum_{j} f(t)\cdot \pi_{j}(t)\cdot t_{j}\cdot \I[t\notin R_{j}] \stepcounter{equation}\tag{\theequation} \label{eq:single bidder value term}\\
&+ \sum_{t\in T} \sum_{j} f(t)\cdot \pi_{j}(t)\cdot \tp_{j}(t_{j})\cdot \I[t\in R_{j}]\stepcounter{equation}\tag{\theequation} \label{eq:single bidder virtual value term}\\
\end{align*}}

Above, $\textsc{Single}$ refers to the bound coming from cases where $t \in R_j$. We name it ``single'' to reference the connection to single-dimensional settings. $\textsc{Non-Favorite}$ refers to the bound coming from cases where $t \notin R_j$, and we name it ``non-favorite'' because this contribution only comes from non-favorite items. We bound \textsc{Single} below, and \textsc{Non-Favorite} differently for unit-demand and additive valuations. 

\begin{lemma}\label{lem:1 SINGLE}
 For any feasible $\pi(\cdot)$, \textsc{Single} $\leq$ \copies.
\end{lemma}
\begin{proof}
Assume $M$ is the mechanism that induces $\pi(\cdot)$. Consider another mechanism $M'$ for the Copies setting, such that for every type profile $t$, $M'$ serves agent $j$ iff $M$ allocates item $j$ in the original setting and $t\in R_{j}$. As $M$ is feasible in the original setting, $M'$ is clearly feasible in the Copies setting. When agent $j$'s type is $t_{j}$, its probability of being served in $M'$ is $\sum_{t_{-j}}f_{-j}(t_{-j})\cdot\pi_{j}(t_{j},t_{-j})\cdot \I[t\in R_{j}]$ for all $j$ and $t_{j}$. Therefore, \textsc{Single} is the ironed virtual welfare achieved by $M'$ with respect to $\tp(\cdot)$. Since the copies setting is a single dimensional setting, the optimal revenue \copies\ equals the maximum ironed virtual welfare, thus no smaller than \textsc{Single}. Note that this proof makes use of the assumption that item values are independent, as otherwise Myerson's theory doesn't apply.
\end{proof}

\paragraph{Upper Bound for a Unit-demand Bidder.}
{As mentioned previously, the bulk of our work is in obtaining a benchmark and properly decomposing it. Now that we have a decomposition, we can use techniques  similar to those of Chawla et al.~\cite{ChawlaHK07,ChawlaHMS10,ChawlaMS10} to approximate each term.} 

\begin{lemma}\label{lem:unit-demand value}
When the types are unit-demand, for any feasible $\pi(\cdot)$, \textsc{Non-Favorite} $\leq$ \copies.
\end{lemma}
\begin{proof}
Indeed, we will prove that \textsc{Non-Favorite} is upper bounded by the revenue of the VCG mechanism in the Copies setting. Define $S(t)$ to be the second largest number in $\{t_{1},\cdots, t_{m}\}$. When the types are unit-demand, the Copies setting is a single item auction with $m$ bidders. Therefore, if we run the Vickrey auction in the Copies setting, the revenue is $\sum_{t\in T} f(t)\cdot S(t)$. If $t\notin R_{j}$, then there exists some $k\neq j$ such that $t_{k}\geq t_{j}$, so $t_{j}\cdot \I[t\in R_{j}]\leq S(t)$ for all $j$. Therefore, 
$\sum_{t\in T} \sum_{j} f(t)\cdot \pi_{j}(t)\cdot t_{j}\cdot \I[t\notin R_{j}]\leq \sum_{t\in T} \sum_{j} f(t)\cdot \pi_{j}(t)\cdot S(t) \leq \sum_{t\in T} f(t)\cdot S(t)$. The last inequality is because the bidder is unit demand, so $\sum_{j}\pi_{j}(t)\leq 1$. 
\end{proof}

Combining Lemma~\ref{lem:1 SINGLE} and Lemma~\ref{lem:unit-demand value}, we recover the result of 
Chawla et al. \cite{ChawlaMS10}:\footnote{This bound combined with~\cite{ChawlaHMS10} recovers the state-of-the-art $4$-approximation via item-pricing.}
\begin{theorem}\label{thm:1 unit-demand}
For a single unit-demand bidder, the optimal revenue is upper bounded by $2$\copies.
\end{theorem}

\paragraph{Upper Bound for an Additive Bidder.}
When the bidder is additive, we need to further decompose \textsc{Non-Favorite} into two terms we call \textsc{Core} and \textsc{Tail}. For simplicity of notation in the proofs that follow, define $r=\srev$. {Again, we remind the reader that most of our work is already done in obtaining our decomposition. The remaining portion of the proof is indeed inspired by prior work of Babaioff et al.~\cite{BabaioffILW14}.\footnote{The resulting 6-approximation is roughly the state of the art - subsequent works have improved the analysis of the same mechanism to guarantee a 5.2-approximation~\cite{MaSL15}.}}

\begin{align*}
\sum_{t\in T} \sum_{j} f(t)&\cdot \pi_{j}(t)\cdot t_{j}\cdot \I[t\notin R_{j}]
\leq \sum_{t\in T} \sum_{j} f(t)\cdot t_{j}\cdot \I[t\notin R_{j}]\\
=& \sum_{j} \sum_{t_{j}> r} f_{j}(t_{j})\cdot t_{j}\cdot \sum_{t_{-j}}f_{-j}(t_{-j})\cdot \I[t\notin R_{j}]
+\sum_{j} \sum_{t_{j}\leq r} f_{j}(t_{j})\cdot t_{j}\cdot\sum_{t_{-j}}f_{-j}(t_{-j})\cdot \I[t\notin R_{j}]\\
\leq& \sum_{j} \sum_{t_{j}> r} f_{j}(t_{j})\cdot t_{j}\cdot \Pr_{t_{-j}\sim D_{-j}}[t\notin R_{j}]\quad (\textsc{Tail})\ +\ \sum_{j} \sum_{t_{j}\leq r} f_{j}(t_{j})\cdot t_{j}\quad (\textsc{Core})
\end{align*}

Before proceeding, let's parse term \textsc{Tail} above (we'll parse $\textsc{Core}$ shortly after). \textsc{Tail} captures contributions to the bound coming from \emph{non-favorite} items whose value is at least \srev. In the term \textsc{Tail}, the main idea is that we should expect $t_j \cdot \Pr[t \notin R_j]$ to be small when $t_j > \srev$. This is because $t_j$ is already quite large, so we should expect the probability that we see another item with even larger value (a necessary condition for $t \notin R_j$) to be quite small. Lemma~\ref{lem:single tail} captures this formally, and makes use of the fact that item values are independent.

\begin{lemma}\label{lem:single tail}
\textsc{Tail} $\leq \srev$.
\end{lemma}
\begin{proof}
(Recall that we define $r = \srev$ for ease of notation in some places). By the definition of $R_{j}$, for any given $t_j$, $$\Pr_{t_{-j}\sim D_{-j}}[t\notin R_{j}]\leq\Pr_{t_{-j}\sim D_{-j}}[\exists k\neq j,\ t_{k}\geq t_{j}].\footnote{If not for issues of tie-breaking, this would be equality.}$$
It is clear that by setting price $t_j$ on each item separately, we can make revenue at least $t_j \cdot \Pr_{t_{-j}\sim D_{-j}}[\exists k \neq j, t_k \geq t_j]$, as the buyer will certainly choose to purchase something at price $t_j$ whenever there is an item she values above $t_j$. So we see that therefore $\srev \geq t_j \cdot \Pr_{t_{-j}\sim D_{-j}}[t\notin R_j]$, for all $t_j$. Thus, \textsc{Tail}  $\leq \srev \cdot \sum_{j} \sum_{t_{j}> r} f_{j}(t_{j})=\sum_{j} \srev \cdot \Pr_{t_{j}\sim D_{j}}[t_{j}>\srev]=$ the revenue of selling each item separately at price $\srev$, which by the same exact reasoning is also $\leq \srev$.
\end{proof}

Now, let's parse $\textsc{Core}$. $\textsc{Core}$ captures contibutions to the bound coming from \emph{non-favorite} items whose value is at most $\srev$. The main idea is that \textsc{Core} is the expected sum of independent random variables, each supported on $[0,\srev]$. So maybe \textsc{Core} = $O(\srev)$, which is great. Or, maybe $\textsc{Core}>>\srev$, in which case it should concentrate (due to being the sum of ``small'' independent random variables). In the latter case, we should expect to have $\brev = \Omega(\textsc{Core})$, which is also great. Lemma~\ref{lem:single core} states this formally, and also makes use of the fact that item values are independent.

\begin{lemma}\label{lem:single core}
If we sell the grand bundle at price $\textsc{Core}-2r$, the bidder will purchase it with probability at least $1/2$. In other words, $\brev\geq {\textsc{Core}\over 2}-r$, 
or $\textsc{Core} \leq 2 \brev + 2 \srev.$
\end{lemma}
\begin{proof}
We will first need a technical lemma (also used in~\cite{BabaioffILW14}, but proved here for completeness).
\begin{lemma}\label{lem:second moment}
Let $x$ be a positive single dimensional random variable drawn from $F$ of finite support,\footnote{The same statement holds for continuous distribution as well, and can be proved using integration by parts.} such that for any number $a$, $a\cdot\Pr_{x\sim F}[x\geq a]\leq \mathcal{B}$ where $\mathcal{B}$ is an absolute constant. Then for any positive number $s$, the second moment of the random variable $x_{s}=x\cdot\I[x\leq s]$ is upper bounded by $2\mathcal{B}\cdot s$.
\end{lemma}
\begin{proof}
Let $\{a_{1},\ldots, a_{\ell}\}$ be the intersection of the support of $F$ and $[0,s]$, and $a_{0}=0$. 
\begin{align*}
\E[x_{s}^{2}]=& \sum_{k=0}^{\ell} \Pr_{x\sim F}(x=a_{k})\cdot a_{k}^{2}\\
=& \sum_{k=1}^{\ell} (a_{k}^{2}-a_{k-1}^{2})\cdot \sum_{d= k}^{\ell} \Pr_{x\sim F}(x=a_{d})\\
\leq& \sum_{k=1}^{\ell} (a_{k}^{2}-a_{k-1}^{2})\cdot \Pr_{x\sim F}[x\geq a_{k}]\\
\leq & \sum_{k=1}^{\ell} 2(a_{k}-a_{k-1})\cdot a_{k}\cdot \Pr_{x\sim F}[x\geq a_{k}]\\
\leq& 2\mathcal{B}\cdot \sum_{k=1}^{\ell} (a_{k}-a_{k-1})\\
\leq & 2\mathcal{B}\cdot s
\end{align*}
The penultimate inequality is because $a_{k}\cdot \Pr_{x\sim F}[x\geq a_{k}]\leq \mathcal{B}$. 
\end{proof}

Now with Lemma~\ref{lem:second moment}, for each $j$ define a new random variable $c_{j}$ based on the following procedure: draw a sample $v_{j}$ from $D_{j}$, if $v_{j}$ lies in $[0,r]$, then $c_{j}=v_{j}$, otherwise $c_{j}=0$. Let $c=\sum_{j} c_{j}$. It is not hard to see that we have $\E[c]=\sum_{j}\sum_{t_{j}\leq r} f_{j}(t_{j})\cdot t_{j}$. Now we are going to show that $c$ concentrates because it has small variance. Since the $c_{j}$'s are independent, $\Var[c]=\sum_{j}\Var[c_{j}]\leq \sum_{j}\E[c_{j}^{2}]$. We will bound each $\E[c_{j}^{2}]$ separately. Let $r_{j}=\max_{x} \{x\cdot \Pr_{t_{j}\sim D_{j}}[t_{j}\geq x]\}$. By Lemma~\ref{lem:second moment}, we can upper bound $\E[c_{j}^{2}]$ by $2r_{j}\cdot r$. On the other hand, it is easy to see that $r=\sum_{j} r_{j}$ (as this is exactly the definition of \srev), so $\Var[c]\leq 2r^{2}$. By the Chebyshev inequality, $$\Pr[c< \E[c]-2r]\leq {\Var[c]\over 4r^{2}}\leq {1\over 2}.$$ Therefore, $$\Pr_{t\sim D}[\sum_{j} t_{j}\geq \E[c]-2r]\geq \Pr[c\geq\E[c]-2r]\geq {1\over 2}.$$ 
So $\brev \geq \frac{\E[c]-2r}{2}$, as we can sell the grand bundle at price $\E[c]-2r$, and it will be purchased with probability at least $1/2$.
\end{proof}

\begin{theorem}
For a single additive bidder, the optimal revenue is $\leq 2\brev+4\srev$.
\end{theorem}
\begin{proof}
Combining Lemma~\ref{lem:1 SINGLE}, \ref{lem:single tail} and \ref{lem:single core}, the optimal revenue is upper bounded by \copies$+\srev+2\brev+2\srev$. It is not hard to see that \copies $ = \srev$, because the optimal auction in the copies setting just sells everything separately. So the optimal revenue is upper bounded by $2\brev+4\srev$.
\end{proof}

%% file: multi_bidder.tex

\section{Multiple Bidders}\label{sec:multi}
In this section, we show how to use the upper bound in Theorem~\ref{thm:virtual welfare ub} to show that deterministic DSIC mechanisms can achieve a constant fraction of the  (randomized) optimal BIC revenue in multi-bidder settings when the bidders valuations are all unit-demand or additive. Before beginning, we remind the reader of some notation from Section~\ref{sec:prelim}: $\vcg(D)$ refers to the revenue of the VCG mechanism when buyers have values drawn from $D$, and $\bvcg$ refers to the revenue of the optimal ``VCG with entry fees'' mechanism. \copies$(D)$ refers to the optimal achievable revenue in the related single-dimensional ``copies'' setting, where each buyer has been split into $m$ different buyers (one for each item). 

Similar to the single bidder case, we first decompose the upper bound (Expression~\ref{eq:ub}) into three components and bound them separately. 
In the last expression  in what follows, we call the first term \textsc{Non-Favorite}, the second term \textsc{Under} and the third term \textsc{Single}. 
We further break \textsc{Non-Favorite} into two  parts, \textsc{Over} and \textsc{Surplus} and bound them separately. The following are the approximation factors we achieve: 

\begin{theorem}\label{thm:multi unit-demand}
For multiple unit-demand bidders, the optimal revenue is upper bounded by $4$\copies.
\end{theorem}

\begin{theorem}\label{thm:multi additive}
For multiple additive bidders, the optimal revenue is upper bounded by $6$\copies $+2\bvcg$.
\end{theorem}

Note that a simple posted-price mechanism achieves revenue \copies$/6$ when all buyers are unit-demand~\cite{ChawlaHMS10, KleinbergW12}, and selling each item separately using Myerson's auction achieves revenue \copies\ when buyers are additive. Therefore, the CHMS/KW~\cite{ChawlaHMS10,KleinbergW12} posted-price mechanism achieves a 24-approximation to the optimal BIC mechanism (previously, it was known to be a 30-approximation), and Yao's approximation ratios~\cite{Yao15} are improved from 69 to 8. {Some parts of the following analysis draw inspiration from prior works of Chawla et al.~\cite{ChawlaHMS10} and Yao~\cite{Yao15}, however, much of the analysis also represents new techniques. In particular, it is worth pointing out that our proof of Theorem~\ref{thm:multi additive} looks similar to our single-bidder case, whereas Yao's original proof required the entirely new machinery of ``$\beta$-adjusted revenue'' and ``$\beta$-exclusive mechanisms.'' Below is our decomposition, first into \textsc{Non-Favorite}, \textsc{Under}, and \textsc{Single}, then further decomposing \textsc{Non-Favorite} into \textsc{Over} and \textsc{Surplus}. Recall that $\wedge$ refers to ``AND'' and $\vee$ refers to ``OR.''} 

\begin{align*}
&\sum_{i}\sum_{t_{i}\in T_{i}} \sum_{j} f_{i}(t_{i})\cdot \pi_{ij}(t_{i})\cdot \Big(t_{ij}\cdot\Pr_{v_{-i}\sim D_{-i}}[t_{i}\notin R^{(v_{-i})}_{j}] +\tp_{ij}(t_{ij})\cdot \Pr_{v_{-i}\sim D_{-i}}[t_{i}\in R^{(v_{-i})}_{j}]\Big)\\
\leq &\sum_{i}\sum_{t_{i}\in T_{i}} \sum_{j} f_{i}(t_{i})\cdot \pi_{ij}(t_{i})\cdot\sum_{v_{-i}\in T_{-i}} t_{ij} f_{-i}(v_{-i})
\cdot\I\Big[\big(\exists k\neq j,\ t_{ik}-P_{ik}(v_{-i})\geq t_{ij}-P_{ij}(v_{-i})\big) \lor\big(t_{ij}<P_{ij}(v_{-i})\big) \Big] \\
+&\sum_{i}\sum_{t_{i}\in T_{i}} \sum_{j} f_{i}(t_{i}) \pi_{ij}(t_{i})\tp_{ij}(t_{ij})\Pr_{v_{-i}\sim D_{-i}}[t_{i}\in R^{(v_{-i})}_{j}]\\
\leq&\sum_{i}\sum_{t_{i}\in T_{i}} \sum_{j} f_{i}(t_{i})\cdot \pi_{ij}(t_{i})\cdot\\
& ~~~~~~~~\sum_{v_{-i}\in T_{-i}} t_{ij} f_{-i}(v_{-i})\cdot\I\Big[\big(\exists k\neq j,\ t_{ik}-P_{ik}(v_{-i})\geq t_{ij}-P_{ij}(v_{-i})\big)\land \big(t_{ij}\geq P_{ij}(v_{-i})\big) \Big]\quad\textsc{(Non-Favorite)}\\
&+\sum_{i}\sum_{t_{i}\in T_{i}} \sum_{j} f_{i}(t_{i})\cdot \pi_{ij}(t_{i})\cdot\sum_{v_{-i}\in T_{-i}} t_{ij} \cdot f_{-i}(v_{-i})\cdot\I[ t_{ij}<P_{ij}(v_{-i})]
\quad\textsc{(Under)}\\
&+\sum_{i}\sum_{t_{i}\in T_{i}} \sum_{j} f_{i}(t_{i})\cdot \pi_{ij}(t_{i})\cdot\tp_{ij}(t_{ij}) \cdot\Pr_{v_{-i}\sim D_{-i}}[t_{i}\in R^{(v_{-i})}_{j}]  \quad\textsc{(Single)}
\end{align*}

\begin{align*}
&\textsc{Non-Favorite}
\leq \sum_{i}\sum_{t_{i}\in T_{i}} \sum_{j} f_{i}(t_{i})\cdot\pi_{ij}(t_{i}) \cdot\sum_{v_{-i}\in T_{-i}} P_{ij}(v_{-i}) f_{-i}(v_{-i})\I[t_{ij}\geq P_{ij}(v_{-i})] \quad (\textsc{Over})\\
&+ \sum_{i}\sum_{t_{i}\in T_{i}} \sum_{j} f_{i}(t_{i})\cdot \pi_{ij}(t_{i})\cdot \\
&~~~~\sum_{v_{-i}\in T_{-i}} \left(t_{ij}-P_{ij}(v_{-i})\right) \cdot f_{-i}(v_{-i})\cdot\I\Big[\big(\exists k\neq j,\ t_{ik}-P_{ik}(v_{-i})\geq t_{ij}-P_{ij}(v_{-i})\big)\land \big(t_{ij}\geq P_{ij}(v_{-i})\big) \Big]\quad (\textsc{Surplus})
\end{align*}

\noindent Before continuing, let's try to parse these five terms:
\begin{itemize}
\item All terms sum over all bidders, all types, and all items, and take the density of that type times the interim probability that bidder receives that item when reporting that type, times some portion of the virtual valuation for that item. 
\item \textsc{Non-Favorite} takes the \textbf{value} for the item, times the probability that it is \emph{not the bidder's favorite item}, as defined in Section~\ref{sec:flow}, when $v_{-i}$ is drawn from $D_{-i}$ (roughly corresponds to items $k$ such that $t_i \in R_j^{(v_{-i})}$ for $j \neq k$, but not perfectly).  
\item \textsc{Under} takes the \textbf{value} for the item, times the probability that the bidder is \emph{not even willing to purchase the item} at the VCG prices defined by $v_{-i}$ drawn from $D_{-i}$ (roughly corresponds to when $t_i \in R_0^{(v_{-i})}$, but not perfectly). 
\item \textsc{Single} takes the \textbf{Myerson Ironed Virtual Value} for the item, times the probability that it \emph{is the bidder's favorite item} (corresponds to items $j$ such that $t_i \in R_j^{(v_{-i})}$). 
\item \textsc{Over} and \textsc{Surplus} split \textsc{Non-Favorite} in the following way:
\begin{itemize}
\item \textsc{Over} replaces the value in \textsc{Non-Favorite} with the VCG price induced by $v_{-i}$ (and also upper bounds some probabilities by $1$). This roughly corresponds to the revenue obtained by VCG (but not perfectly). 
\item \textsc{Surplus} replaces the value in \textsc{Non-Favorite} with (value - VCG price induced by $v_{-i}$), and roughly corresponds to the bidder's utility for participating in the VCG auction (but not perfectly). 
\item Observe that value = VCG price + (value - VCG price), so this is indeed a decomposition of \text{Non-Favorite}. 
\end{itemize}
\end{itemize}

The plan of attack is as follows: \textsc{Single} will be handled the same way as in Section~\ref{sec:single}. \textsc{Surplus} will be handled similarly to \textsc{Non-Favorite} from Section~\ref{sec:single}, and both parts yield the same approximation guarantees as their single-bidder counterparts. That leaves \textsc{Under} and \textsc{Over}, which we will show each contribute at most an additional \copies, and account for the ``plus two'' in transitioning from single-bidder to multi-bidder bounds. We now proceed to address these terms formally, beginning with \textsc{Surplus}.\\

\noindent\textbf{Analyzing \textsc{Surplus} for Unit-demand Bidders:} The proof of this lemma  is similar in spirit to \Cref{lem:unit-demand value}.

\begin{lemma}\label{lem:multi unit-demand value}
	When the types are unit-demand, for any feasible $\pi(\cdot)$, \textsc{Surplus} $\leq$ \copies.
\end{lemma}
\begin{proof}
		Indeed, we will prove that \textsc{Surplus} is bounded above by the revenue of the VCG mechanism in the Copies setting. For any $i$ define $S_{i}(t_{i},v_{-i})$ to be the second largest number in $\{t_{i1}-P_{i1}(v_{-i}),\cdots, t_{im}-P_{im}(v_{-i})\}$. Now consider running the VCG mechanism on type profile $(t_{i},v_{-i})$. An agent $(i,j)$ is served in the VCG mechanism in the Copies setting, iff item $j$ is allocated to $i$ in the VCG mechanism in the original setting, which is equivalent to saying $t_{ij}-P_{ij}(v_{-i})\geq 0$ and $t_{ij}-P_{ij}(v_{-i})\geq t_{ik}-P_{ik}(v_{-i})$ for all $k$. The Copies setting is single-dimensional, therefore any agent's payment is her threshold bid. For agent $(i,j)$, her threshold bid is $P_{ij}(v_{-i})+\max\{0, \max_{k\neq j} t_{ik}-P_{ik}(v_{-i})\}$ which is at least $S_{i}(t_{i},v_{-i})$. On the other hand, for any $i$, whenever $\exists j',\ t_{ij'}-P_{ij'}(v_{-i})\geq 0$, there exists some $j_{i}$ such that $(i,j_{i})$ is served in the VCG mechanism. Combining the two conclusions above, we show that on any profile $(t_{i},v_{-i})$, the payment in the VCG mechanism collected from agents in $\{(i,j)\}_{j\in[m]}$ is at least $S_{i}(t_{i},v_{-i})\cdot\I[\exists j',\ t_{ij'}-P_{ij'}(v_{-i})\geq 0]$. So the total revenue of the VCG Copies mechanism is at least:
		
		$$\sum_{i}\sum_{(t_{i},v_{-i})\in T_i} f(t_{i},v_{-i})\cdot S_{i}(t_{i},v_{-i})\cdot\I[\exists j',\ t_{ij'}-P_{ij'}(v_{-i})\geq 0].$$
		
		Next we argue for any $j$ and $(t_{i},v_{-i})$, the following inequality holds. 
\begin{align*}(t_{ij}-P_{ij}(v_{-i})) \cdot \I\Big[&\big(\exists k\neq j,\ t_{ik}-P_{ik}(v_{-i})\geq t_{ij}-P_{ij}(v_{-i}) \geq 0 \Big]\\ 
&\leq S_{i}(t_{i},v_{-i})\cdot\I[\exists j',\ t_{ij'}-P_{ij'}(v_{-i})\geq 0] \stepcounter{equation}\tag{\theequation} \label{ieq:surplus}\end{align*} 
		We only need to consider the case when the LHS is non-zero. In that case, the RHS has value $S_{i}(t_{i},v_{-i})$, and also there exists some $k$ such that $t_{ik}-P_{ik}(v_{-i})\geq t_{ij}-P_{ij}(v_{-i})$, so $t_{ij}-P_{ij}(v_{-i})\leq S_{i}(t_{i},v_{-i})$. 
		
		So now we can rewrite \textsc{Surplus} and upper bound it with the revenue of the VCG mechanism in the Copies setting.
\begin{align*}
		&\sum_{i}\sum_{t_{i}\in T_{i}} \sum_{j} f_{i}(t_{i})\cdot \pi_{ij}(t_{i})\sum_{v_{-i}\in T_{-i}} (t_{ij}-P_{ij}(v_{-i}))\cdot f_{-i}(v_{-i})\cdot\I\Big[\exists k\neq j,\ t_{ik}-P_{ik}(v_{-i})\geq t_{ij}-P_{ij}(v_{-i})\geq 0\Big]\\
		&= \sum_{i}\sum_{(t_{i},v_{-i})\in T_i} f(t_{i},v_{-i}) \sum_{j} \pi_{ij}(t_{i})\cdot(t_{ij}-P_{ij}(v_{-i})) \cdot \I\Big[\exists k\neq j,\ t_{ik}-P_{ik}(v_{-i})\geq t_{ij}-P_{ij}(v_{-i})\geq 0 \Big]\\
		&\leq  \sum_{i}\sum_{(t_{i},v_{-i})\in T_i} f(t_{i},v_{-i})\sum_{j} \pi_{ij}(t_{i})\cdot S_{i}(t_{i},v_{-i}) \cdot\I[\exists j',\ t_{ij'}-P_{ij'}(v_{-i})\geq 0]\qquad\text{(Inequality~(\ref{ieq:surplus}))}\\
		&\leq  \sum_{i} \sum_{(t_{i},v_{-i})\in T_i} f(t_{i},v_{-i}) \cdot S_{i}(t_{i},v_{-i}) \cdot\I[\exists j',\ t_{ij'}-P_{ij'}(v_{-i})\geq 0] \quad (\sum_{j} \pi_{ij}(t_{i})\leq 1\ \forall i, t_{i})
		\end{align*}		
		The last line is upper bounded by the revenue of the VCG mechanism in the Copies setting by our work above, which is clearly upper bounded by \copies.
\end{proof}

\paragraph{Analyzing \textsc{Surplus} for Additive Bidders:}
Similar to the single bidder case, we will again break the term \textsc{Surplus} into the \textsc{Core} and the \textsc{Tail}, and analyze them separately. Before we proceed, we first define the cutoffs. Let $r_{ij}(v_{-i})=\max_{x\geq P_{ij}(v_{-i})} \{x\cdot \Pr_{t_{ij\sim D_{ij}}}[t_{ij}\geq x]\}$. {The observant reader will notice that this is bidder $i$'s ex-ante payment for item $j$ in Ronen's single-item mechanism~\cite{Ronen01} conditioned on other bidders types being $v_{-i}$, but this connection is not necessary to understand the proof}. Further let $r_{i}(v_{-i})=\sum_{j} r_{ij}(v_{-i})$, $r_{i}=\E_{v_{-i}\sim D_{-i}}[r_{i}(v_{-i})]$ and $r=\sum_{i} r_{i}$, the expected revenue of running Ronen's mechanism separately for each item {(again, the connection to Ronen's mechanism is not necessary to understand the proof)}. We first bound \textsc{Tail} and \textsc{Core}, using arguments similar to the single item case
(Lemmas \ref{lem:single tail} and \ref{lem:single core}), 

\begin{align*}
&\textsc{Surplus} \leq \sum _{i} \sum_{v_{-i}\in T_{-i}}f_{-i}(v_{-i})\sum_{j} \sum_{t_{ij}\geq P_{ij}(v_{-i})}  f_{ij}(t_{ij})\cdot (t_{ij}-P_{ij}(v_{-i}))\cdot\sum_{t_{i,-j}\in T_{i,-j}}f_{i,-j}(t_{i,-j})\cdot\\
&~~~~~~~~~~~~~~~~~~~~~~~~~~~~~~~~~~~~~~~~~~~~~~~~~~~~~~~~~~~~~~~~~~~~~~~~~~~~~~~~~~~~~~~~~~~~\I[\exists k\neq j,\ t_{ik}-P_{ik}(v_{-i})\geq t_{ij}-P_{ij}(v_{-i})]\\
&= \sum _{i} \sum_{v_{-i}\in T_{-i}}f_{-i}(v_{-i})\sum_{j} \sum_{t_{ij}\geq P_{ij}(v_{-i})} f_{ij}(t_{ij})\cdot \\
&~~~~~~~~~~~~~~~~~~~~~~~~~~~~~~~~~~~~~~~~~~(t_{ij}-P_{ij}(v_{-i}))\cdot \Pr_{t_{i,-j}\sim D_{i,-j}}[\exists k\neq j, t_{ik}-P_{ik}(v_{-i})\geq t_{ij}-P_{ij}(v_{-i})]\\
&\leq \sum _{i} \sum_{v_{-i}\in T_{-i}}f_{-i}(v_{-i})\sum_{j} \sum_{t_{ij}> P_{ij}(v_{-i})+r_{i}(v_{-i})} f_{ij}(t_{ij})\cdot \\
& ~~~~~~~~~~~~~~~~~~~~~~~~~~~~~~~~~~~~~~~~~~(t_{ij}-P_{ij}(v_{-i}))\cdot\Pr_{t_{i,-j}\sim D_{i,-j}}[\exists k\neq j,t_{ik}-P_{ik}(v_{-i})\geq t_{ij}-P_{ij}(v_{-i})] \quad (\textsc{Tail})\\
&+ \sum _{i} \sum_{v_{-i}\in T_{-i}}f_{-i}(v_{-i})\sum_{j} \sum_{t_{ij}\in [P_{ij}(v_{-i}), P_{ij}(v_{-i})+r_{i}(v_{-i})]} f_{ij}(t_{ij})\cdot (t_{ij}-P_{ij}(v_{-i}))\quad (\textsc{Core})
\end{align*}

\begin{lemma}\label{lem:multitail}
	$\textsc{Tail}\leq r$.
\end{lemma}
\begin{proof} First, by union bound \begin{align*} &\Pr_{t_{i,-j}\sim D_{i,-j}}[\exists k\neq j,\ t_{ik}-P_{ik}(v_{-i})\geq t_{ij}-P_{ij}(v_{-i})] 
\leq &\sum_{k\neq j} \Pr_{t_{ik}\sim D_{ik}}[\ t_{ik}-P_{ik}(v_{-i})\geq t_{ij}-P_{ij}(v_{-i})].
\end{align*}
 By the definition of $r_{ik}(v_{-i})$, {we certainly have $r_{ik}(v_{-i}) \geq (P_{ik}(v_{-i}) + t_{ij} - P_{ij}(v_{-i}))\cdot \Pr_{t_{ik}\sim D_{ik}}[\ t_{ik}-P_{ik}(v_{-i})\geq t_{ij}-P_{ij}(v_{-i})]$, so we can also derive:}
\begin{align*}
 &\Pr_{t_{ik}\sim D_{ik}}[\ t_{ik}-P_{ik}(v_{-i})\geq t_{ij}-P_{ij}(v_{-i})]
 \leq {r_{ik}(v_{-i})\over P_{ik}(v_{-i})+t_{ij}-P_{ij}(v_{-i})}\leq {r_{ik}(v_{-i})\over t_{ij}-P_{ij}(v_{-i})}.\end{align*}
 Using these two inequalities, we can upper bound \textsc{Tail}:

		\begin{align*}
&\sum_{i}\sum_{v_{-i}\in T_{-i}}f_{-i}(v_{-i})\sum_{j} \sum_{t_{ij}> P_{ij}(v_{-i})+r_{i}(v_{-i})} f_{ij}(t_{ij})
		\cdot \sum_{k\neq j} r_{ik}(v_{-i})\\
		\leq& \sum_{i}\sum_{v_{-i}} f_{-i}(v_{-i})\cdot\sum_{j} r_{i}(v_{-i}) \cdot \sum_{t_{ij}> P_{ij}(v_{-i})+r_{i}(v_{-i})} f_{ij}(t_{ij})\\
		\leq & \sum_{i}\sum_{v_{-i}} f_{-i}(v_{-i})\sum_{j} r_{ij}(v_{-i})\quad \text{(Definition of $r_{ij}(v_{-i})$})\\
		=& r\end{align*}
		\end{proof}

\begin{lemma}\label{lem:multicore}
	$\bvcg\geq {\textsc{Core}\over 2}-r$. In other words, $2r+2\bvcg\geq \textsc{Core}$.
\end{lemma}

\begin{proof}
		Fix any $v_{-i}\in T_{-i}$, let $t_{ij}\sim D_{ij}$, define two new random variables $$b_{ij}(v_{-i}) = (t_{ij}-P_{ij}(v_{-i}))\I[t_{ij}\geq P_{ij}(v_{-i})]$$  and $$c_{ij}(v_{-i})=b_{ij}(v_{-i})\I[b_{ij}(v_{-i}) \leq r_{i}(v_{i})].$$ Clearly, $c_{ij}(v_{-i})$ is supported on $[0,r_{i}(v_{-i})]$. Also, we have \begin{align*}&\E_{t_{ij} \sim D_{ij}}[c_{ij}(v_{-i})] 
		=  \sum_{t_{ij}\in [P_{ij}(v_{-i}), P_{ij}(v_{-i})+r_{i}(v_{-i})]} f_{ij}(t_{ij})\cdot (t_{ij}-P_{ij}(v_{-i})).\end{align*}
		 So we can rewrite \textsc{Core} as 
		$$\sum _{i} \sum_{v_{-i}\in T_{-i}}f_{-i}(v_{-i})\sum_{j} \E[c_{ij}(v_{-i})].$$
		
		Now we will describe a VCG mechanism with per bidder entry fee. Define an entry fee function for bidder $i$ depending on $v_{-i}$ as $e_i(v_{-i})=\sum_{j} \E[c_{ij}(v_{-i})]-2r_{i}(v_{-i})$. We will show that for any $i$ and other bidders types $v_{-i}\in T_{-i}$, bidder $i$ accepts the entry fee $e_{i}(v_{-i})$ with probability at least $1/2$. Since bidders are additive, the VCG mechanism is exactly $m$ separate Vickrey auctions, one for each item. So $P_{ij}(v_{-i})=\max_{\ell\neq i} \{v_{\ell j}\}$, and for any set of $S$, its Clarke Pivot price for $i$ to receive set $S$ is $\sum_{j\in S} P_{ij}(v_{-i})$. 
		
		That also means $\sum_{j} b_{ij}(v_{-i})$ is the random variable that represents bidder $i$'s utility in the VCG mechanism when other bidders bids are $v_{-i}$. If we can prove 
		$\Pr[\sum_{j} b_{ij}(v_{-i})\geq e_{i}(v_{-i})]\geq 1/2$ for all $v_{-i}$, then we know bidder $i$ accepts the entry fee with probability at least $1/2$.
		
		It is not hard to see for any nonnegative number $a$, \begin{align*} &a\cdot\Pr[b_{ij}(v_{-i})\geq a]
		\leq (a+P_{ij}(v_{-i}))\cdot \Pr[t_{ij}\geq a+P_{ij}(v_{-i})]\leq r_{ij}(v_{-i}).\end{align*} Therefore, because each $c_{ij}(v_{-i}) \in [0,r_i(v_{-i})]$, by Lemma~\ref{lem:second moment} we can again bound the second moment as: $\E[c_{ij}(v_{-i})^{2}]\leq 2r_{i}(v_{-i})r_{ij}(v_{-i})$. Since $c_{ij}$'s are independent, \begin{align*}&\Var[\sum_{j} c_{ij}(v_{-i})]=\sum_{j}\Var[ c_{ij}(v_{-i})]
		\leq \sum_{j} \E[c_{ij}(v_{-i})^{2}]\leq 2r_{i}(v_{-i})^{2}.\end{align*}
		
		By Chebyshev inequality, we know \begin{align*}&\Pr[\sum_{j} c_{ij}(v_{-i})\leq \sum_{j} \E[c_{ij}(v_{-i})]-2r_{i}(v_{-i})]
		\leq {\Var[\sum_{j} c_{ij}(v_{-i})]\over 4r_{i}(v_{-i})^{2}}\leq 1/2.\end{align*}
		Therefore, as $b_{ij}(v_{-i}) \geq c_{ij}(v_{-i})$, we can conclude:
$$\Pr[\sum_{j} b_{ij}(v_{-i})\geq e_{i}(v_{-i})]\geq 1/2$$
		
		So the entry fee is accepted with probability at least $1/2$ for all $i$ and $v_{-i}$. So: 

\begin{align*} 
&\bvcg \geq {1\over 2}\sum_{i}\sum_{v_{-i}\in T_{-i}}f_{-i}(v_{-i}) \big(\E[c_{ij}(v_{-i})]-2r_{i}(v_{-i})\big)= {\textsc{Core}\over 2} - r.\end{align*}
\end{proof}

\notshow{
	\begin{theorem}\label{thm:1 unit-demand}
		For a single unit-demand bidder, the optimal revenue is upper bounded by $2OPT^{copies}$.
	\end{theorem}
	\begin{proof}
		Combining Lemma~\ref{lem:1 SINGLE} and Lemma~\ref{lem:unit-demand value}.
	\end{proof}
	
	\begin{align*}
	&\sum_{i}\sum_{j}\int_{T_{i}}f_{i}(t_{i}) \cdot\pi_{ij}(t_{i}) \cdot \sum_{v_{-i}\in T_{-i}: t_{i}\in R_{j}\text{ for flow }\lambda^{i}_{v_{-i}}} f_{-i}(v_{-i}) \Phi_{ij}^{v_{-i}}(t_{i}) d t_{i}\\
	=&\sum_{i}\sum_{j}\sum_{v_{-i}\in T_{-i}} f_{-i}(v_{-i}) \int_{T_{i,-j}}f_{i,-j}(t_{i,-j})\int_{x_{l}^{(v_{-i}, t_{i,-j})}}^{+\infty} f_{ij}(t_{ij}) \cdot\pi_{ij}(t_{i})\Phi_{ij}^{v_{-i}}(t_{i}) d t_{ij} d t_{i,-j}\\
	\leq &\sum_{i}\sum_{j}\sum_{v_{-i}\in T_{-i}} f_{-i}(v_{-i}) \int_{T_{i,-j}}f_{i,-j}(t_{i,-j})\int_{x_{l}^{(v_{-i}, t_{i,-j})}}^{+\infty} f_{ij}(t_{ij}) \cdot\Phi_{ij}^{v_{-i}}(t_{i}) d t_{ij} d t_{i,-j}\\
	\leq & \sum_{i}\sum_{j}\sum_{v_{-i}\in T_{-i}} f_{-i}(v_{-i}) \int_{T_{i,-j}}f_{i,-j}(t_{i,-j})\cdot r_{ij}|v_{-i,j} d t_{i,-j}\ (Corollary~\ref{cor:bound virtual value})\\
	= &r
	\end{align*}
	
	\section{Bounding the first term}
	First notice that $\Pr[(\exists k\neq j, t_{ik}-P_{ik}(t_{-i})\geq t_{ij}-P_{ij}(t_{-i})) \lor (t_{ij}-P_{ij}(t_{-i})< 0)]=\Pr[t_{ij}< P_{ij}(t_{-i})]+\Pr[(t_{ij}\geq S_{ij}) \land (\exists k\neq j, t_{ik}-P_{ik}(t_{-i})\geq t_{ij}-P_{ij}(t_{-i})) )]$. Let's rewrite the first term 
	\begin{align*}
	&\sum_{i}\sum_{j}\int_{T_{i}}f_{i}(t_{i})\cdot\pi_{ij}(t_{i}) \cdot t_{ij}\Pr[(\exists k\neq j, t_{ik}-P_{ik}(t_{-i})\geq t_{ij}-P_{ij}(t_{-i})) \lor (t_{ij}-P_{ij}(t_{-i})< 0)|t_{i}] dt_{i}\\
	=& \sum_{i}\sum_{j}\int_{T_{i}}f_{i}(t_{i})\cdot\pi_{ij}(t_{i}) \cdot t_{ij} \Pr[t_{ij}< P_{ij}(t_{-i})|t_{i}] d t_{i}\qquad (1)\\
	+&\sum_{i}\sum_{j}\int_{T_{i}}f_{i}(t_{i})\cdot\pi_{ij}(t_{i}) \cdot t_{ij} \Pr[(t_{ij}\geq S_{ij}) \land (\exists k\neq j, t_{ik}-P_{ik}(t_{-i})\geq t_{ij}-P_{ij}(t_{-i})) )|t_{i}] d t_{i}\qquad (2)
	\end{align*} 
	
	\subsection{Bounding (1)}
	Let me first bound (1) using the revenue of the Vickrey auction.
	\begin{align*}
	&(1)=\sum_{i}\sum_{j}\int_{T_{ij}} f_{ij}(t_{ij})\cdot t_{ij}\cdot\Pr[t_{ij}< P_{ij}(t_{-i})|t_{ij}] \int_{T_{i,-j}} f_{i,-j}(t_{i,-j})\cdot \pi_{ij}(t_{i}) d t_{i,-j} d t_{ij} \qquad (independence)
	\end{align*}
	
	Let $\pi_{ij}(t_{ij})=\int_{T_{i,-j}} f_{i,-j}(t_{i,-j})\cdot \pi_{ij}(t_{i}) d t_{i,-j}$, clearly this is a feasible reduced form for item $j$. Now the above formula becomes 
	\begin{align*}
	&\sum_{i}\sum_{j}\int_{T_{ij}} f_{ij}(t_{ij})\cdot \pi_{ij}(t_{ij})\cdot t_{ij}\cdot\Pr[t_{ij}< P_{ij}(t_{-i})|t_{ij}] dt_{ij}\\
	\leq &\sum_{i}\sum_{j}\int_{T_{j}} f_{j}(t_{j})\cdot \pi_{ij}(t_{ij})\cdot m_{j}d t_{j}\\
	\leq &\sum_{j}m_{j} \sum_{i}\int_{T_{j}} f_{j}(t_{j})\cdot \pi_{ij}(t_{ij})d t_{j}\\
	\leq &\sum_{j} m_{j}\\
	\leq & SREV\\
	\end{align*}
	
	\subsection{Bounding (2)}
	First, we relax (2) by replacing all $\pi_{ij}(t_{i})$ with $1$. Now it becomes $$\sum_{i}\sum_{j}\int_{T_{i}}f_{i}(t_{i})\cdot t_{ij} \Pr[(t_{ij}\geq P_{ij}(t_{-i})) \land (\exists k\neq j, t_{ik}-P_{ik}(t_{-i})\geq t_{ij}-P_{ij}(t_{-i})) )|t_{i}] d t_{i}.$$
	Let $g_{j}(\cdot)$ be the density function of $P_{ij}(t_{-i})$ when $t_{-ij}$ is drawn from $D_{-ij}$. Now we can rewrite the formula above as
	\begin{align*}
	&\sum_{i}\sum_{j}\int_{T_{i}}f_{i}(t_{i})\cdot t_{ij} \int _{P_{ij}(t_{-i})\leq t_{ij}} g_{j}(P_{ij}(t_{-i}))\Pr[\exists k\neq j, t_{ik}-P_{ik}(t_{-i})\geq t_{ij}-P_{ij}(t_{-i})) |t_{i},P_{ij}(t_{-i})]d P_{ij}(t_{-i}) d t_{i}\\
	=&\sum_{i}\sum_{j}\int_{T_{ij}}f_{ij}(t_{ij})\cdot t_{ij}\int_{P_{ij}(t_{-i})\leq t_{ij}} g_{j}(P_{ij}(t_{-i}))\int_{T_{i,-j}} f_{i,-j}(t_{i,-j})\Pr[\exists k\neq j, t_{ik}-P_{ik}(t_{-i})\geq t_{ij}-P_{ij}(t_{-i})) |t_{i},P_{ij}(t_{-i})]dt_{i,-j} d P_{ij}(t_{-i}) d t_{ij}\\
	= &\sum_{i}\sum_{j}\int_{T_{ij}}f_{ij}(t_{ij})\cdot t_{ij}\int_{P_{ij}(t_{-i})\leq t_{ij}} g_{j}(P_{ij}(t_{-i})) \Pr[\exists k\neq j, t_{ik}-P_{ik}(t_{-i})\geq t_{ij}-P_{ij}(t_{-i})) |t_{ij},P_{ij}(t_{-i})] d P_{ij}(t_{-i}) d t_{ij}\\
	= & \sum_{i}\sum_{j}\int_{0}^{+\infty} g_{j}(P_{ij}(t_{-i}))\int_{t_{ij}\geq P_{ij}(t_{-i})}f_{ij}(t_{ij})\cdot t_{ij} \Pr[\exists k\neq j, t_{ik}-P_{ik}(t_{-i})\geq t_{ij}-P_{ij}(t_{-i})) |t_{ij},P_{ij}(t_{-i})]  d t_{ij} d P_{ij}(t_{-i})\\
	= & \sum_{i}\sum_{j}\int_{0}^{+\infty} g_{j}(P_{ij}(t_{-i}))\int_{t_{ij}\geq P_{ij}(t_{-i})}f_{ij}(t_{ij})\cdot [P_{ij}(t_{-i})+(t_{ij}-P_{ij}(t_{-i}))] \Pr[\exists k\neq j, t_{ik}-P_{ik}(t_{-i})\geq t_{ij}-P_{ij}(t_{-i})) |t_{ij},P_{ij}(t_{-i})]  d t_{ij} d P_{ij}(t_{-i})\\
	\leq & \sum_{i}\sum_{j}\int_{0}^{+\infty} g_{j}(P_{ij}(t_{-i}))\int_{t_{ij}\geq P_{ij}(t_{-i})}f_{ij}(t_{ij})\cdot P_{ij}(t_{-i}) d t_{ij} d_{P_{ij}(t_{-i})}\\
	+ &\sum_{i}\sum_{j}\int_{0}^{+\infty} g_{j}(P_{ij}(t_{-i}))\int_{t_{ij}\geq P_{ij}(t_{-i})}f_{ij}(t_{ij})\cdot (t_{ij}-P_{ij}(t_{-i})) \Pr[\exists k\neq j, t_{ik}-P_{ik}(t_{-i})\geq t_{ij}-P_{ij}(t_{-i})) |t_{ij},P_{ij}(t_{-i})]  d t_{ij} d P_{ij}(t_{-i}) \qquad (3)\\
	= & V + (3)
	\end{align*}
	
	Now we only need to bound (3), this turns out to be very similar to the single bidder case.
	
	\subsection{Bounding (3)}
	We separate the sum into two parts, (i) $t_{ij}\in[P_{ij}(t_{-i}), P_{ij}(t_{-i})+r_{i}]$; and (ii) $t_{ij}>P_{ij}(t_{-i})+r_{ij}$, and bound them using different methods.
	\begin{align*}
	&(3)\\
	=& \sum_{i}\sum_{j}\int_{0}^{+\infty} g_{j}(P_{ij}(t_{-i}))\int_{P_{ij}(t_{-i})}^{P_{ij}(t_{-i})+r_{i}}f_{ij}(t_{ij})\cdot (t_{ij}-P_{ij}(t_{-i})) \Pr[\exists k\neq j, t_{ik}-P_{ik}(t_{-i})\geq t_{ij}-P_{ij}(t_{-i})) |t_{ij},P_{ij}(t_{-i})]  d t_{ij} d P_{ij}(t_{-i})\\
	+& \sum_{i}\sum_{j}\int_{0}^{+\infty} g_{j}(P_{ij}(t_{-i}))\int_{P_{ij}(t_{-i})+r_{i}}^{+\infty}f_{ij}(t_{ij})\cdot (t_{ij}-P_{ij}(t_{-i})) \Pr[\exists k\neq j, t_{ik}-P_{ik}(t_{-i})\geq t_{ij}-P_{ij}(t_{-i})) |t_{ij},P_{ij}(t_{-i})]  d t_{ij} d P_{ij}(t_{-i}) \\
	\leq & \sum_{i}\sum_{j}\int_{0}^{+\infty} g_{j}(P_{ij}(t_{-i}))\int_{P_{ij}(t_{-i})}^{P_{ij}(t_{-i})+r_{i}}f_{ij}(t_{ij})\cdot (t_{ij}-P_{ij}(t_{-i}))  d t_{ij} d P_{ij}(t_{-i}) \qquad (A)\\
	+ & \sum_{i}\sum_{j}\int_{0}^{+\infty} g_{j}(P_{ij}(t_{-i}))\int_{P_{ij}(t_{-i})+r_{i}}^{+\infty}f_{ij}(t_{ij})\cdot (t_{ij}-P_{ij}(t_{-i})) (\sum_{k\neq j} \Pr[t_{ik}-P_{ik}(t_{-i})\geq t_{ij}-P_{ij}(t_{-i})) |t_{ij},P_{ij}(t_{-i})])  d t_{ij} d P_{ij}(t_{-i}) \qquad (B)
	\end{align*}
	
	We will first bound (B). Notice that $\Pr[t_{ik}-P_{ik}(t_{-i})\geq t_{ij}-P_{ij}(t_{-i})) |t_{ij},P_{ij}(t_{-i})]= \Pr[t_{ik}\geq t_{ij}-P_{ij}(t_{-i})+P_{ik}(t_{-i})) |t_{ij},P_{ij}(t_{-i})]$, so therefore I can consider the following auction. Whenever bidder $i$ has the highest value for item $k$, I ask her to pay the second highest plus $t_{ij}-P_{ij}(t_{-i})$. The expected revenue is greater than $\Pr[t_{ik}\geq t_{ij}-P_{ij}(t_{-i})+P_{ik}(t_{-i})) |t_{ij},P_{ij}(t_{-i})]\cdot (t_{ij}-P_{ij}(t_{-i}))$, but is less than $r_{ik}$. Therefore, $\Pr[t_{ik}-P_{ik}(t_{-i})\geq t_{ij}-P_{ij}(t_{-i})) |t_{ij},P_{ij}(t_{-i})]\leq \frac{r_{ik}}{t_{ij}-P_{ij}(t_{-i})}$. Therefore, we have 
	\begin{align*}
	&(B)\\
	\leq &\sum_{i}\sum_{j}\int_{0}^{+\infty} g_{j}(P_{ij}(t_{-i}))\int_{P_{ij}(t_{-i})+r_{i}}^{+\infty}f_{ij}(t_{ij}) (\sum_{k\neq j} r_{ik})d t_{ij} d P_{ij}(t_{-i})\\
	\leq & \sum_{i}\sum_{j}\int_{0}^{+\infty} g_{j}(P_{ij}(t_{-i}))r_{i} \int_{P_{ij}(t_{-i})+r_{i}}^{+\infty}f_{ij}(t_{ij}) d t_{ij} d P_{ij}(t_{-i})\\
	= &\sum_{i}\sum_{j}\int_{0}^{+\infty} g_{j}(P_{ij}(t_{-i}))r_{i} (1- F_{ij}(P_{ij}(t_{-i})+r_{i}))d P_{ij}(t_{-i})\\
	\leq& \sum_{i}\sum_{j} r_{ij}\\
	= &r
	\end{align*}
	The last inequality is because $\int_{0}^{+\infty} g_{j}(P_{ij}(t_{-i}))r_{i} (1- F_{ij}(P_{ij}(t_{-i})+r_{i}))d P_{ij}(t_{-i})$ is less than the revenue of the following mechanism: when $i$ has highest value for item $j$, ask her to pay $P_{ij}(t_{-i})+r_{i}$, and $i$'s expected payment in this mechanism is less than $r_{ij}$.

	The last step is to bound (A). 
	
	\subsection{Bounding (A)}
	Let $c_{ij}$ be a random variable such that $c_{ij}= (r_{ij}-P_{ij}(t_{-i}))\mathbb{I}\{r_{ij}\in [P_{ij}(t_{-i}),P_{ij}(t_{-i})+r_{i}]\}$. (A)$= \sum_{i}\sum_{j}\mathbb{E}[c_{ij}]$. Let $c_{i}= \sum_{j} c_{ij}$. If $c_{i}$ concentrates, then it's easy to extract a constant fraction of $\mathbb{E}[c_{i}]$, by charging entree fee to bidder $i$. We call bidder$i$ good if $c_{i}$ concentrates. Now the key is to argue that if (A) is large comparing to $r$, then the sum of $\mathbb{E}[c_{i}]$ of good bidders is almost (A). 
	
	Again, we will use Chebyshev to argue concentration. $Var[c_{i}]=\sum_{j} Var[c_{ij}]\leq \sum_{j} \mathbb{E}[c_{ij}^{2}]$. Now we show $\mathbb{E}[c_{ij}^{2}]\leq 2r_{i}r_{ij}$.
	\begin{align*}
	&\mathbb{E}[c_{ij}^{2}]=\int_{0}^{+\infty} g_{j}(P_{ij}(t_{-i}))\int_{P_{ij}(t_{-i})}^{P_{ij}(t_{-i})+r_{i}}f_{ij}(t_{ij})\cdot (t_{ij}-P_{ij}(t_{-i}))^{2}  d t_{ij} d P_{ij}(t_{-i})\\
	=& \int_{0}^{+\infty} g_{j}(P_{ij}(t_{-i})) 2\int_{P_{ij}(t_{-i})}^{P_{ij}(t_{-i})+r_{i}}(t_{ij}-P_{ij}(t_{-i}))(F_{ij}(P_{ij}(t_{-i})+r_{i})-F_{ij}(t_{ij})dt_{ij}d_{P_{ij}(t_{-i})}\\
	\leq & 2\int_{0}^{+\infty} g_{j}(P_{ij}(t_{-i})) \int_{P_{ij}(t_{-i})}^{P_{ij}(t_{-i})+r_{i}}t_{ij}(1-F_{ij}(t_{ij})dt_{ij}d_{P_{ij}(t_{-i})}\\
	\leq & 2\int_{0}^{+\infty} g_{j}(S_{j)} \int_{P_{ij}(t_{-i})}^{P_{ij}(t_{-i})+r_{i}}r_{ij}|P_{ij}(t_{-i})dt_{ij} d_{P_{ij}(t_{-i})}\\
	= & 2 r_{i} \int_{0}^{+\infty} g_{j}(S_{j)} r_{ij}|P_{ij}(t_{-i}) dP_{ij}(t_{-i})\\
	= & 2 r_{i} r_{ij}
	\end{align*}
	Therefore $Var[c_{i}]\leq 2r_{i}^{2}$. So if we set the entree fee for bidder $i$ to be $\mathbb{E}[c_{i}]-2r_{i}$, by Chebyshev bidder $i$ will buy with probability at least $1/2$. So the expected revenue is at least (A)$/2 - r$. (A) is upper bounded by the revenue of two Ronen plus two entree fee. 
	
	So totally the virtual welfare is upper bounded by revenue of four Ronen, one Vickrey, one SREV and two entree fee. The approximation factor is $8$.}

\noindent\textbf{Analyzing \textsc{Single}, \textsc{Over} and \textsc{Under}:}
First we consider \textsc{Single}, which is similar to  \Cref{lem:1 SINGLE}. 

\begin{lemma}\label{lem:SINGLE}
	For any feasible $\pi(\cdot)$, \textsc{Single} $\leq$ \copies.
\end{lemma}
\begin{proof}
		Assume $M$ is the ex-post allocation rule that induces $\pi(\cdot)$. Consider another ex-post allocation rule $M'$ for the copies setting, such that for every type profile $t$, if $M$ allocates item $j$ to bidder $i$ in the original setting then $M'$ serves agent $(i,j)$ with probability $\Pr_{v_{-i}\sim D_{-i}}[t_{i}\in R_{j}^{(v_{-i})}]$. As $M$ is feasible in the original setting, $M'$ is clearly feasible in the Copies setting. When agent $(i,j)$ has type $t_{ij}$, her probability of being served in $M'$ is 
\begin{align*}
&\sum_{t_{i,-j}}f_{i,-j}(t_{i,-j})\cdot\pi_{ij}(t_{ij},t_{i,-j})\cdot\Pr_{v_{-i}\sim D_{-i}}[(t_{ij},t_{i,-j})\in R_{j}^{(v_{-i})}] 
\end{align*}
for all $j$ and $t_{ij}$.
		Therefore, \textsc{Single} is the ironed virtual welfare achieved by $M'$ with respect to $\tp(\cdot)$. Since the copies setting is a single dimensional setting, the optimal revenue \copies\ equals the maximum ironed virtual welfare, thus no smaller than \textsc{Single}.
\end{proof}


Next, we move onto \textsc{Over}. Recall that the terms in \textsc{Over} are VCG prices times an indicator that the bidder's value exceeds the VCG prices. So we should hope to be able to cover \textsc{Over} with some VCG-like mechanism. We begin with the following technical propositions:

\begin{proposition}\label{prop:copies}
Let $\pi(\cdot)$ be any reduced form of a BIC mechanism in the original setting. Define $$\Pi_{ij}(t_{ij}) = \mathbb{E}_{t_{i,-j}\sim D_{i,-j}}[\pi_{ij}(t_i)].$$ Then $\Pi_{ij}(t_{ij})$ is monotone in $t_{ij}$.
\end{proposition}
\begin{proof}
In fact, for all $t_{i,-j}$, we must have $\pi_{ij}(\cdot,t_{i,-j})$ monotone increasing in $t_{ij}$. Assume for contradiction that this were not the case, and let $t_{ij} < t'_{ij}$ with $\pi_{ij}(t_{ij},t_{i,-j})> \pi_{ij}(t'_{ij},t_{i,-j})$. Then $(t_{ij},t_{i,-j}), (t'_{ij},t_{i,-j})$ form a 2-cycle that violates cyclic monotonicity. This is because both types value all items except for $j$ exactly the same. 

To expand a bit for readers not familiar with cyclic monotonicity: observe that $t_{ij} < t'_{ij}$ but $\pi_{ij}(t_{ij},t_{i,-j}) > \pi_{ij}(t'_{ij}, t_{i, -j})$ implies that

$$(t_{ij},t_{i,-j}) \cdot \pi_{i}(t_{ij}, t_{i,-j}) + (t'_{ij},t_{i,-j}) \cdot \pi_{i}(t'_{ij}, t_{i,-j}) <(t'_{ij},t_{i,-j}) \cdot \pi_{i}(t_{ij}, t_{i,-j}) + (t_{ij},t_{i,-j}) \cdot \pi_{i}(t'_{ij}, t_{i,-j})$$
$$\Rightarrow (t_{ij},t_{i,-j}) \cdot \pi_{i}(t_{ij}, t_{i,-j})-(t_{ij},t_{i,-j}) \cdot \pi_{i}(t'_{ij}, t_{i,-j}) <(t'_{ij},t_{i,-j}) \cdot \pi_{i}(t_{ij}, t_{i,-j})-(t'_{ij},t_{i,-j}) \cdot \pi_{i}(t'_{ij}, t_{i,-j}).$$

This directly implies that no matter what prices are set for $p_i(t_{ij},t_{i,-j})$ and $p_i(t'_{ij}, t_{i,-j})$, if bidder $i$ with type $(t_{ij},t_{i,-j})$ is happy to tell the truth, then type $(t'_{ij}, t_{i,-j})$ strictly prefers to lie and report $(t_{ij}, t_{i,-j})$ than tell the truth. 
\end{proof}

\begin{proposition}\label{prop:pi with reserve}
	For any $v\in T$, any $\pi(\cdot)$ that is a reduced form of some BIC mechanism, \begin{align*} &\textsc{OPT}^{\textsc{Copies}} \geq \sum_{i}\sum_{t_{i}\in T_{i}} \sum_{j} f_{i}(t_{i})\cdot \pi_{ij}(t_{i})\cdot P_{ij}(v_{-i})\cdot\I[t_{ij}\geq P_{ij}(v_{-i})] .\end{align*}
	\end{proposition}
\begin{proof} Recall from Proposition~\ref{prop:copies} that every BIC interim form $\pi(\cdot)$ in the original setting corresponds to a monotone interim form in the copies setting, $\Pi(\cdot)$. Let $M$ be any (possibly randomized) allocation rule that induces $\Pi(\cdot)$, and $p(\cdot)$ a corresponding price rule (wlog we can let $(M,p)$ be ex-post IR). Consider the following mechanism instead: on input $t$, first run $(M,p)$ to (possibly randomly) determine a set of potential winners. Then, if $(i,j)$ is a potential winner, offer $(i,j)$ service at price $\max\{p_{ij}(t),P_{ij}(v_{-i}))$. Whenever $(i,j)$ is a potential winner, $t_{ij}\geq p_{ij}(t)$. It is clear that in the event that $(i,j)$ is a potential winner, and $t_{ij} \geq P_{ij}(t_{-i})$, $(i,j)$ will accept the price and pay at least $P_{ij}(v_{-i})$. Therefore, for any $t$ as long as $(i,j)$ is served in $M$, then the payment from $(i,j)$ in the new proposed mechanism is at least $P_{ij}(v_{-i})\I[t_{ij}\geq P_{ij}(v_{-i})]$. That means the total revenue of the new mechanism is at least $\sum_{i}\sum_{t_{i}\in T_{i}} \sum_{j} f_{i}(t_{i})\cdot \pi_{ij}(t_{i})\cdot P_{ij}(v_{-i})\cdot\I[t_{ij}\geq P_{ij}(v_{-i})]$, which is upper bounded by \copies.
	\end{proof}

\begin{lemma}\label{lem:over}
	$\textsc{Over}\leq$ \copies.
\end{lemma}
\begin{proof}
	This can be proved by rewriting \textsc{Over} and then applying Proposition~\ref{prop:pi with reserve}.
	\begin{align*}
\textsc{Over}=&\sum_{i}\sum_{t_{i}\in T_{i}} \sum_{j} f_{i}(t_{i})\cdot \pi_{ij}(t_{i})\cdot\sum_{v\in T} P_{ij}(v_{-i}) f(v)\I[t_{ij}\geq P_{ij}(v_{-i})]\\
= &\sum_{v\in T} f(v)\sum_{i}\sum_{t_{i}\in T_{i}} \sum_{j} f_{i}(t_{i})\cdot \pi_{ij}(t_{i})\cdot P_{ij}(v_{-i}) \cdot \I[t_{ij}\geq P_{ij}(v_{-i})]\\
	\leq& \sum_{v\in T} f(v)\cdot \textsc{OPT}^{\textsc{copies}}=\textsc{OPT}^{\textsc{copies}}
	\end{align*}
\end{proof}

Finally, we move on to \textsc{Under}. When there is only one bidder, \textsc{Under} is always $0$. 
Here, \textsc{Under} $\leq$ \copies, and turns out to be the trickiest part to bound. Recall that \textsc{Under} contains terms that are (non-favorite) values times indicators that these values \emph{do not} exceed the VCG prices. So the high-level hope is that the reason the VCG price for bidder $i$ to receive item $j$ exceeds $t_{ij}$ is because someone else is paying at least $t_{ij}$ for something, and we might hope to be able to come up with a clever charging argument. At a high level, this is indeed the plan, but the proof approach doesn't clearly map onto this intuition. 
We apply Proposition~\ref{prop:vcg reserve copies} (below) once for each type profile $t$, 
using the allocation of this mechanism on type profile $t$ to specify 
$(i_j,j)$ and  let $x_j  = t_{i_jj}$. 
Then taking the convex combination of the RHS of Proposition~\ref{prop:vcg reserve copies} for all profiles $t$ with multipliers $f(t)$ gives \textsc{Under}$\leq$ \copies.

\begin{proposition}\label{prop:vcg reserve copies}
Let $\{(i_{j},j)\}_{j\in S\subseteq[m]}$ be a feasible allocation in the copies setting. For all choices $x_1,\ldots,x_m\geq 0$, \copies $\geq \sum_{v\in T} f(v) \cdot \sum_{j\in S} x_j\cdot \I[P_{i_j j}(v_{-i_{j}}) > x_j]$.
\end{proposition}

\begin{proof}

Before beginning the proof of Proposition~\ref{prop:vcg reserve copies}, we will need the following definition and theorem due to Gul and Stacchetti~\cite{GulS99}.
\begin{definition} 
Let $W_T(S)$ be the maximum attainable welfare using only bidders in $T$ and items in $S$. \end{definition}

\begin{theorem}\label{thm:GS}(\cite{GulS99})
If all bidders in  $T$ have gross substitute valuations, then $W_T(\cdot)$ is a submodular function.
\end{theorem}

Now with Theorem~\ref{thm:GS}, consider in the Copies setting the VCG mechanism with lazy reserve $x_j$ for each copy $(i,j)$. Specifically, we will first solicit bids, then find the max-welfare allocation and call all $(i,j)$ who get allocated temporary winners. Then, if $(i,j)$ is a temporary winner, $(i,j)$ is given the option to receive service for the maximum of their Clarke pivot price in the Copies setting and $x_j$. It is clear that in this mechanism, whenever any agent $(i,j)$ receives service, the price she pays is at least $x_j$. Also, it is not hard to see that this is a truthful mechanism (for the Copies): for all other fixed bids, copy $(i,j)$ can report a bid exceeding the maximum of $x_j$ and their Clarke pivot price, or not. If they report a higher bid, they will receive service and pay the maximum of their Clarke pivot price and $x_j$. If they report a smaller number, they remain unserved. It's clear that bidding the Copy's true value is always optimal. Next, we argue for any $v\in T$ and $j\in S$, whenever $P_{i_j j}(v_{-i_{j}}) > x_j$, there exists some $i$ such that $(i,j)$ is served in the mechanism above.

By the definition of Clarke pivot price, we know $$P_{i_j j}(v_{-i_{j}})=W_{[n]-\{i_{j}\}}([m])-W_{[n]-\{i_{j}\}}([m]-\{j\}).$$ First, we show that if item $j$ is allocated to some bidder $i$ in the max-welfare allocation in the original setting then $v_{ij}\geq P_{ij}(v_{-i})$. Assume $S'$ to be the set of items allocated to bidder $i$. Since the VCG mechanism is truthful, the utility for winning set $S'$ is better than winning set $S'-\{j\}$: 
\begin{align*}\sum_{k\in S'}& v_{ik} -(W_{[n]-\{i\}}([m])-W_{[n]-\{i\}}([m]-S'))\\
&\geq \sum_{k\in S'-\{j\}} v_{ik} -(W_{[n]-\{i\}}([m]) -W_{[n]-\{i\}}([m]-S'+\{j\})).\end{align*}
Rearranging the terms, we get
\begin{align*}
v_{ij}
\geq &W_{[n]-\{i\}}([m]-S'+\{j\})-W_{[n]-\{i\}}([m]-S')\\
\geq & W_{[n]-\{i\}}([m])-W_{[n]-\{i\}}([m]-\{j\})) \quad\text(Theorem~\ref{thm:GS})\\
 = &P_{ij}(v_{-i}).
\end{align*}

Now we still need to argue that whenever $P_{i_j j}(v_{-i_j}) > x_j$, item $j$ is always allocated in the max-welfare allocation to some bidder $i$ with $v_{ij}\geq x_{j}$.
\begin{enumerate}
\item If agent $(i_{j},j)$ is a temporary winner, $$v_{i_{j}j}\geq P_{i_j j}(v_{-i_{j}}) > x_j.$$ Therefore, agent $(i_{j},j)$ will accept the price.
\item If agent $(i_{j},j)$ is not a temporary winner, let $S'$ be the set of items that are allocated to bidder $i_{j}$ in the welfare maximizing allocation in the original setting. Since $W_{[n]-\{i_{j}\}}([m]-S')-W_{[n]-\{i_{j}\}}([m]-S'-\{j\})\geq W_{[n]-\{i_{j}\}}([m])-W_{[n]-\{i_{j}\}}([m]-\{j\})= P_{i_{j}j}(v_{-i_{j}})$ (by Theorem~\ref{thm:GS}), and $P_{i_j j}(v_{-i_j})>x_{j}$, the following are true: (i) item $j$ is awarded to some bidder $i\neq i_j$ in the welfare maximizing allocation, because otherwise $W_{[n]-\{i_{j}\}}([m]-S')$ will have the same value as $W_{[n]-\{i_{j}\}}([m]-S'-\{j\})$; (ii) $v_{ij}>x_{j}$ because  $$v_{ij}\geq W_{[n]-\{i_{j}\}}([m]-S')- W_{[n]-\{i_{j}\}}([m]-S'-\{j\}) = P_{i_j j}(v_{-i_j})> x_j.$$
\end{enumerate}

So now we can conclude that for any $j\in S$ there is certainly some $i$ such that $(i,j)$ is served whenever $P_{i_j j} > x_j$, 
and therefore the revenue of this mechanism in the Copies setting is at least $\sum_{v\in T} f(v) \cdot \sum_{j\in S} x_j\cdot \I[P_{i_j j}(v_{-i_{j}}) > x_j]$, 
which is exactly the same as the sum in the proposition statement.
\end{proof}

\begin{lemma}\label{lem:multi under}
\textsc{Under} $\leq$ \copies.
\end{lemma}
\begin{proof} The idea is to interpret \textsc{Under} as the revenue of the following mechanism: let $M$ be the mechanism that induces $\pi(\cdot)$. Sample $t$ from $D$, let $S$ be the set of agents that will be served in $M$ for type profile $t$ in the copies setting. Use $t_{ij}$ to be the reserve price for $j$ if $(i,j)\in S$, and use the mechanism in Proposition~\ref{prop:vcg reserve copies}.

First, the inner sum $$\sum_{v_{-i}\in T_{-i}} t_{ij}\cdot f_{-i}(v_{-i})\cdot\I[ t_{ij}<P_{ij}(v_{-i})]$$ only depends on $t_{i}$, so the maximum of \textsc{Under} is achieved by a $\pi(\cdot)$ induced by some deterministic mechanism. Wlog, we consider $\pi(\cdot)$ is induced by a deterministic mechanism whose ex-post allocation rule is $x(\cdot)$. Let us rewrite \textsc{Under} using $x(\cdot)$:
\begin{align*}
\sum_{i}\sum_{t_{i}\in T_{i}}& \sum_{j} f_{i}(t_{i})\cdot \pi_{ij}(t_{i})\cdot \sum_{v_{-i}\in T_{-i}} t_{ij}\cdot f_{-i}(v_{-i})\cdot\I[ t_{ij}<P_{ij}(v_{-i})]\\
=&\sum_{t\in T}f({t}) \sum_{i}\sum_{j} x_{ij}(t)\cdot t_{ij}\cdot\sum_{v\in T} f(v)\cdot\I[ t_{ij}<P_{ij}(v_{-i})]\\
=&\sum_{t\in T}f({t})\cdot \sum_{v\in T} f(v)\sum_{i}\sum_{j} x_{ij}(t)\cdot t_{ij}\cdot\I[ t_{ij}<P_{ij}(v_{-i})]\\
\leq & \sum_{t\in T}f({t})\cdot \textsc{OPT}^\textsc{copies}= \textsc{OPT}^\textsc{copies}
\end{align*} 

The penultimate inequality is because if we let $\{(i_{j},j)\}_{j\in S}$ be the set of agents such that by $x_{i_{j}j}(t)=1$, then  \begin{align*}&\sum_{v\in T} f(v)\sum_{i}\sum_{j} x_{ij}(t)\cdot t_{ij}\cdot\I[ t_{ij}<P_{ij}(v_{-i})]
= \sum_{v\in T} f(v) \cdot \sum_{j\in S} x_j\cdot \I[P_{i_j j}(v_{-i_{j}}) > x_j],\end{align*} and by Proposition~\ref{prop:vcg reserve copies}, this is upper bounded by \copies.
\end{proof}

Combining the above lemmas now yields our theorems:
 
\begin{prevproof}{Theorem}{thm:multi unit-demand}	
Combine Lemmas~\ref{lem:multi unit-demand value},~\ref{lem:SINGLE},~\ref{lem:over} and~\ref{lem:multi under}.
\end{prevproof}	

\begin{prevproof}{Theorem}{thm:multi additive}	
Combining Lemmas~\ref{lem:multitail},~\ref{lem:multicore},~\ref{lem:SINGLE},~\ref{lem:over} and~\ref{lem:multi under}, we get the optimal revenue is upper bounded by $$3\textsc{OPT}^{\textsc{Copies}}+3r+2\bvcg.$$ Since \copies is exactly the revenue of selling each item separately optimally using Myerson's auction, and $r$ is the revenue of some mechanism (Ronen's) that sells the items separately, we have \copies $\geq r$, proving the statement.
\end{prevproof}



%% file: general.tex

\section{Duality Theory Beyond Additive Bidders}\label{sec:general}
In this section we provide a statement of our duality theory that holds much more generally than when bidders are unit-demand or additive. The technical ideas are exactly the same as in Section~\ref{sec:duality} and just require updated notation. 
\vspace{.1in}

\noindent\textbf{Buyer Valuations.} In this section, we will consider buyers with arbitrary valuation functions for subsets of items. That is, buyer $i$ has some valuation function $t_i(\cdot)$ that takes as input a set of items and outputs a value. Buyer $i$'s type is drawn from some distribution $D_i$, and $D = \times_i D_i$ is the joint distribution over profiles of buyer types. We define $\mathcal{F}$ to be a set system over $[n]\times[m]$ that describes all feasible allocations.
\vspace{.1in}

\noindent\textbf{Implicit Forms.} The implicit form of an auction stores for all bidders $i$, and pairs of types $t_i$, $t'_i$, what is the expected value that that agent $i$ will receive when her real valuation function is $t_i$, but she reports $t'_i$ to the mechanism instead (over the randomness in the mechanism and randomness in other agents' reported types, assuming they come from $D_{-i}$) as $\pi_{i}(t_i,t'_i)$. We say that an implicit form is \emph{feasible} if there exists some feasible mechanism (that selects an outcome in $\mathcal{F}$ with probability $1$) that matches the expectations promised by the implicit form. If $\polytope$ is defined to be the set of all feasible implicit forms, it is easy to see (and shown in~\cite{CaiDW13b}, for instance) that $\polytope$ is closed and convex. Note that implicit forms are computed over the same randomness as reduced forms, but store directly the value that bidder $i$ receives for having type $t_i$ and reporting $t'_i$, instead of indirectly via interim probabilities.
\vspace{.1in}

We begin by writing the LP for revenue maximization in this more general setting (\Cref{fig:LPRevenueG}). 
 To proceed, we'll again introduce a variable $\lambda_i(t,t')$ for each of the BIC constraints,
and take the partial Lagrangian of LP~\ref{fig:LPRevenueG} by Lagrangifying all BIC constraints. 
The theory of Lagrangian multipliers tells us that the solution to  LP~\ref{fig:LPRevenueG} is equivalent to the primal variables solving the partially Lagrangified dual (\Cref{fig:LagrangianG}). 

\begin{definition}
Let $\L(\lambda, \pi, p)$ be a the partial Lagrangian defined as follows:
\begin{align*}
 \L(\lambda, \pi, p)=\sum_{i=1}^{n} \left(\sum_{t_i \in T_i} f_{i}(t_{i})\cdot p_i(t_i)+\sum_{t_{i}\in T_{i}}\sum_{t_{i}'\in T_i^{+}} \lambda_{i}(t_{i},t_{i}')\cdot \Big(\pi_i(t_{i},t_i)-\pi_i({t_i,t_{i}'})-\big(p_{i}(t_{i})-p_{i}(t_{i}')\big)\Big)\right)\stepcounter{equation}\tag{\theequation} \label{eq:primal lagrangianG}
\end{align*}
\begin{align*}
=\sum_{i=1}^n \sum_{t_{i}\in T_{i}} p_{i}(t_{i})\Big(f_{i}(t_{i})+&\sum_{t_{i}'\in T_{i}} \lambda_{i}(t_{i}',t_{i})-\sum_{t_{i}'\in T_{i}^{+}} \lambda_{i}(t_{i},t_{i}')\Big)\\
&+\sum_{i=1}^n\sum_{t_{i}\in T_{i}} \Big(\sum_{t_{i}'\in T_{i}^{+}}\lambda_{i}(t_{i},t_{i}')\cdot t_i(X_i(t_i))-\sum_{t'_{i}\in T_{i}}t_{i}'\cdot \lambda_{i}(t_{i}',t_{i})\cdot t'_i(X_i(t_i)) \Big)\stepcounter{equation}\tag{\theequation} \label{eq:dual lagrangianG}
\end{align*}
\end{definition}

In Equation~\eqref{eq:dual lagrangianG}, we use $X_i(t_i)$ to denote the the random set awarded to bidder $i$ when reporting type $t_i$ to the mechanism. That is, $ t_i(X_i(t'_i)):=\pi_i(t_i, t'_i) $. 

\begin{figure}[ht]
\colorbox{MyGray}{
\begin{minipage}{0.98\textwidth} {
\noindent\textbf{Variables:}
\begin{itemize}
\item $p_i(t_i)$, for all bidders $i$ and types $t_i \in T_i$, denoting the expected price paid by bidder $i$ when reporting type $t_i$ over the randomness of the mechanism and the other bidders' types.
\item $\pi_{i}(t_i,t'_i)$, for all bidders $i$, and types $t_i ,t'_i \in T_i$, denoting the expected value that bidder $i$ receives when her real type is type $t_i$ but reports $t'_i$, over the randomness of the mechanism and the other bidders' types.
\end{itemize}
\textbf{Constraints:}
\begin{itemize}
\item $\pi_i(t_i,t_i) - p_i(t_i) \geq \pi_i(t_i,t'_i) - p_i(t'_i) $, for all bidders $i$, and types $t_i \in T_i, t'_i \in T_i^+$, guaranteeing that the implicit form mechanism $({\pi},{p})$ is BIC and BIR.
\item ${\pi} \in \polytope$, guaranteeing ${\pi}$ is feasible.
\end{itemize}
\textbf{Objective:}
\begin{itemize}
\item $\text{Maximize:} \sum_{i=1}^{n} \sum_{t_i \in T_i} f_{i}(t_{i})\cdot p_i(t_i)$, the expected revenue.\\
\end{itemize}}
\end{minipage}}
\caption{A Linear Program (LP) for Revenue Optimization.}
\label{fig:LPRevenueG}
\end{figure}

\begin{figure}[ht]
\colorbox{MyGray}{
\begin{minipage}{0.98\textwidth} {
\noindent\textbf{Variables:}
\begin{itemize}
\item $\lambda_i(t_{i},t_{i}')$ for all $i,t_{i}\in T_{i},t_{i}' \in T_i^{+}$, the Lagrangian multipliers for Bayesian IC constraints.
\end{itemize}
\textbf{Constraints:}
\begin{itemize}
\item $\lambda_i(t_{i},t_{i}')\geq 0$ for all $i,t_{i}\in T_{i},t_{i}' \in T_i^{+}$, guaranteeing that the Lagrangian multipliers are non-negative. 
\end{itemize}
\textbf{Objective:}
\begin{itemize}
\item $\text{Minimize:} \max_{\pi\in\polytope, p} \L(\lambda, \pi, p)$.\\
\end{itemize}}
\end{minipage}}
\caption{Partial Lagrangian of the Revenue Maximization LP.}
\label{fig:LagrangianG}
\end{figure}

Lemma~\ref{lem:useful dual} immediately holds in this setting as well, and the proof is identical. That is, a feasible dual solution is still useful if and only if it induces a flow in the same graph. We will define virtual valuation functions in essentially the same way, just updating notation.

\begin{definition}[Virtual Value Function]\label{def:virtual valueG}
For each $\lambda$, we define a corresponding virtual value function $\Phi(\cdot)$, such that for every bidder $i$, every type $t_{i}\in T_{i}$, $\Phi_{i}(t_{i})(\cdot)=t_{i}(\cdot)-{1\over f_{i}(t_{i})}\sum_{t_{i}'\in T_{i}} \lambda_{i}(t_{i}',t_{i})(t_{i}'(\cdot)-t_{i}(\cdot)).$ That is, $\Phi_i(t_i)$ is a function that takes as input sets of items and outputs a value. For any set of items $S$, $\Phi_i(t_i)(S) = t_{i}(S)-{1\over f_{i}(t_{i})}\sum_{t_{i}'\in T_{i}} \lambda_{i}(t_{i}',t_{i})(t_{i}'(S)-t_{i}(S)).$
\end{definition}

We can now state the proper generalization of Theorem~\ref{thm:revenue less than virtual welfare} in this general setting. The proof is identical to that of Theorem~\ref{thm:revenue less than virtual welfare} and we omit it. In the theorem statement, $X_i(t_i)$ again denotes the random set allocated to bidder $i$ when reporting type $t_i$, so that $\pi_i(t_i, t'_i) = t_i(X_i(t'_i))$. 

\begin{theorem}[Virtual Welfare $\geq$ Revenue]\label{thm:revenue less than virtual welfareG}
Let $\lambda$ be any useful dual solution and $M = (\pi,p)$ any BIC mechanism. Then the revenue of $M$ is $\leq$ the virtual welfare of $\pi$ w.r.t. the virtual value function $\Phi(\cdot)$ corresponding to $\lambda$. That is:
$$\sum_{i=1}^{n} \sum_{t_i \in T_i} f_{i}(t_{i})\cdot p_i(t_i)\leq \sum_{i=1}^{n} \sum_{t_{i}\in T_{i}} f_{i}(t_{i})\cdot\Phi_{i}(t_{i})(X_i(t_i)).$$
Equality holds if and only if for all $i, t, t'$ such that $\lambda_i(t, t') > 0$, the BIC constraint for bidder $i$ between $t$ and $t'$ binds in $M$ (that is, bidder $i$ with type $t$ is indifferent between reporting $t$ and $t'$). Furthermore, let $\lambda^{*}$ be the optimal dual variables and $M^{*}=(\pi^{*},p^{*})$ be the revenue optimal BIC mechanism, then the expected virtual welfare with respect to $\Phi^{*}$ (induced by $\lambda^{*}$) under  $\pi^{*}$ equals the expected revenue of $M^{*}${, and $$\pi^* \in \argmax_{\pi \in \polytope}\left\{\sum_{i=1}^n \sum_{t_i \in T_i} f_i(t_i) \cdot \Phi^*_i(t_i)(X_i(t_i))\right\}.$$}
\end{theorem}